\documentclass{article}
\usepackage{moreverb}
\usepackage{url}
\usepackage{amsmath}
\usepackage{amsthm}
\usepackage{tikz}
\usepackage{booktabs}
\usepackage{subfig}
\usepackage{graphicx}
\usepackage{slashbox}

\usepackage[disable]{todonotes}
\newcommand{\todoC}[1]{\todo[author=Clem, inline,color=red!25]{#1}}

\newtheorem{theorem}{Theorem}
\newtheorem{proposition}{Proposition}
\newtheorem{corollary}{Corollary}

\begin{document}

\title{Generic algorithms for scheduling applications on heterogeneous multi-core platforms\protect\footnote{An extended abstract of a preliminary version of this paper has been presented in Euro-Par 2017~\cite{europarPaper}}}

\author{
	Marcos Amaris \and Giorgio Lucarelli \and Cl\'ement Mommessin \and Denis Trystram\footnote{Institute of Mathematics and Statistics, University of S\~ao Paulo, S\~ao Paulo, Brazil. amaris@ime.usp.br\newline Univ. Grenoble Alpes, CNRS, Inria, LIG, F-38000 Grenoble, France. \{giorgio.lucarelli, clement.mommessin, denis.trystram\}@inria.fr}
}

\maketitle

\abstract{
We study the problem of executing an application represented by a precedence task graph on a parallel machine composed of standard computing cores and accelerators.
Contrary to most existing approaches, we distinguish the allocation and the scheduling phases and we mainly focus on the allocation part of the problem: choose the most appropriate type of computing unit for each task.
We address both off-line and on-line settings and design generic scheduling approaches.
In the first case, we establish strong lower bounds on the worst-case performance of a known approach based on Linear Programming for solving the allocation problem.
Then, we refine the scheduling phase and we replace the greedy List Scheduling policy used in this approach by a better ordering of the tasks.
Although this modification leads to the same approximability guarantees, it performs much better in practice.
We also extend this algorithm to more types of computing units, achieving an approximation ratio which depends on the number of different types.
In the on-line case, we assume that the tasks arrive in any, not known in advance, order which respects the precedence relations and the scheduler has to take irrevocable decisions about their allocation and execution.
In this setting, we propose the first on-line scheduling algorithm which takes into account precedences.
Our algorithm is based on adequate rules for selecting the type of processor where to allocate the tasks and it achieves a constant factor approximation guarantee if the ratio of the number of CPUs over the number of GPUs is bounded.
Finally, all the previous algorithms for hybrid architectures have been experimented on a large number of simulations built on actual libraries.
These simulations assess the good practical behavior of the algorithms with respect to the state-of-the-art solutions, whenever these exist, or baseline algorithms.
}

\section{Introduction}\label{sec:Intro}

The parallel and distributed platforms available today become more and more \textit{heterogeneous}.
Such heterogeneous architectures, composed of several kind of computing units, have a growing impact on performance in high-performance computing.
Hardware accelerators, such as General Purpose Graphical Processing Units (in short GPUs)~\cite{Lee}, are often used in conjunction with multiple computing units (CPUs) on the same chip sharing the same common memory.
As an instance of this, the number of platforms of the TOP500 equipped with accelerators has significantly increased during the last years~\cite{top500.org}.
In the future it is expected that the nodes of such platforms will be even more diverse than today: they will be composed of fast computing nodes, hybrid computing nodes mixing general purpose units with accelerators, I/O nodes, nodes specialized in data analytics, etc.
The interconnect of a huge number of such nodes will also lead to more heterogeneity.
Using heterogeneous platforms would lead to better performances through the use of more appropriate resources depending on the computations to perform, but it has a cost in terms of code development and more complex resource management.

In this work, we present efficient algorithms for scheduling an application represented by a precedence task graph on hybrid and heterogeneous computing resources.
We are interested in designing generic approaches for efficiently implementing parallel applications where the scheduling is not explicitly part of the application.
In this way, the code is portable and can be adapted to the next generation of machines.\\

\noindent\textbf{Underlying architecture.}\\
We consider a hybrid multi-core machine composed of two different types of processors. The machine is composed of two sets, each containing identical processors of a same type.
An application consists of tasks that are linked by precedence relations.
Each task is characterized by two processing times depending on which type of processor it is assigned to.
We assume that an exact estimation of both these processing times is available to the scheduler.
This assumption can be justified by several existing models to estimate the execution times of tasks~\cite{Amaris:NCA:2016}.
In several applications we always observe an acceleration of the tasks if they are executed on a GPU compared to their execution on a CPU.
However, we consider the more general case where the relation between the two processing times can differ for different tasks.
This work focuses on the analysis of the qualitative behavior induced by heterogeneity since it may be assumed that the computations dominate local shared memory costs.
Thus, no memory assignment or overhead for data management are considered, nor communication times between the shared memory and the processors, or between two processors of different type.
Without loss of generality, we denote in the following by CPU and GPU the two types of processors.

As the application developers are mainly looking for performance, the objective of a scheduler is usually to minimize the completion time of the last finishing task, which is one of the most commonly studied objectives~\cite{Drozdowski:2009:SPP}.
In an heterogeneous context, minimizing the makespan of an application corresponds to minimizing the maximum between the makespan of the tasks assigned on each set of processor types.\\

\noindent\textbf{Definition and notations.}\\
We consider a parallel application which should be scheduled on $m$ identical CPUs and $k$ identical GPUs.
Without loss of generality, we assume that $m \geq k$.
The application is represented by a Directed Acyclic Graph $G=(V,E)$ whose nodes correspond to sequential tasks and arcs correspond to precedence relations among the tasks.
We denote by $\mathcal{T}$ the set of all tasks.
The execution of a task needs a different amount of time if it is performed by a CPU or by a GPU.
Let $\overline{p_j}$ (resp. $\underline{p_j}$) be the processing time of a task $T_j$ if it is executed on any CPU (resp. GPU).
Given a schedule $S$, we denote by $C_j$ the completion time of a task $T_j$ in $S$.
In any feasible schedule, for each arc $(i,j) \in E$, the task $T_j$ cannot be executed before the completion of $T_i$.
We say that $T_i$ is a \emph{predecessor} of $T_j$ and we denote by $\Gamma^-(T_j)$ the set of all predecessors of $T_j$.
Similarly, we say that $T_j$ is a \emph{successor} of $T_i$ and we denote by $\Gamma^{+}(T_i)$ the set of all successors of $T_i$.
We call \emph{descendant} of $T_j$ each task $T_i$ for which there is a path from $j$ to $i$ in $G$.

The objective is to create a feasible non-preemptive schedule of minimum makespan.
In other words, we seek a schedule that respects the precedence constraints among tasks, does not interrupt their execution and minimizes the completion time of the last task, i.e., $C_{\max}=\max_{T_j}\{C_j\}$.
Extending the three-fields notation for scheduling problems introduced by Graham, this problem can be denoted as $(CPU,GPU) \mid prec \mid C_{\max}$.\\

\noindent\textbf{Contributions and outline.}\\
In this paper we study the above problem on both off-line and on-line settings.
The goal is to design algorithms through a solid theoretical analysis that can be practically implemented in actual systems.

Contrarily to most existing approaches (see for example~\cite{HEFT}), we propose to address the problem 
by separately focusing on the following two phases:
\begin{itemize}
\item \emph{allocation}: each task is assigned to a type of resources, either CPU or GPU;
\item \emph{scheduling}: each task is assigned to a specific pair of resource and time interval respecting the decided allocation as well as the precedence constraints.
\end{itemize}

In off-line mode, we aim to study the two phases separately motivated by the fact that there are strong lower bounds on the approximability of known single-phase algorithms.
For example, the approximation ratio of the well-known Heterogeneous Earliest Finish Time (HEFT) algorithm~\cite{HEFT} cannot be better than $\Omega(\frac{m}{k^2})$ when $k \leq \sqrt{m}$ (Section~\ref{sec:LowerBounds}). On the other hand, it can be easily shown that List Scheduling policies have arbitrarily large approximation ratio, even if we consider some enhanced order of tasks, like prioritizing the task of the largest acceleration.

The two-phases approach has been used by Kedad-Sidhoum et al.~\cite{LP6} where a linear program (which we call Heterogeneous Linear Program or simply HLP) in conjunction with a rounding have been proposed for the allocation phase, while the greedy Earliest Starting Time (EST) policy has been applied to schedule the tasks.
This algorithm, called HLP-EST, achieves an approximation ratio of 6 and we show in Section~\ref{sec:LowerBounds} that this ratio is tight.
In fact, our worst-case example does not depend on the scheduling policy applied in the second phase.

Based on this negative result, we propose to revisit both phases and design an off-line and an on-line algorithm.
In Section~\ref{ssec:offlineSetting}, we replace the EST policy in HLP-EST by a specific ordering of tasks combined to a classical List Scheduling. The task ordering is based on both the allocation decisions taken during the first phase (liner program) and the critical path. This refined algorithm, denoted by HLP-OLS, preserves the tight approximation ratio of 6 and achieves good practical performances.

In Section~\ref{ssec:onlineSetting}, we first present three greedy rules which can be used to decide the allocation.
Although these rules are of low complexity, a desired property in practice, they are only based on the relation between the processing times of a task and neither consider the schedule created up to now nor look to the future precedence relations that define the critical path.
For these reasons, they cannot guarantee any approximation ratio.
However, a more enhanced set of rules that takes into account the actual schedule can lead to an algorithm of competitive ratio $\Theta(\sqrt{\frac{m}{k}})$ in the on-line context where the tasks arrive in any order that respects the precedence constraints, and the scheduler has to take irrevocable decisions for their execution at the time of their arrival.
This is the first on-line algorithm when precedence constraints are considered in the hybrid context.

In Section~\ref{sec:generalization} we propose an extension of HLP-EST and its analysis for the case where $Q\geq2$ types of identical processors are available.
We show that this algorithm has a tight approximation ratio of $Q(Q+1)$.

An experimental evaluation of the discussed off-line and on-line algorithms is presented in Section~\ref{sec:Experiments}. 
In the off-line setting, experiments showed that the new scheduling method based on HLP (HLP-OLS) outperforms HLP-EST on both contexts with 2 or 3 resource types with an average improvement of 10\%. Comparisons with HEFT showed that HLP-OLS offered similar performances with an improvement of 2\% on average for 2 resource types, but a performance decrease of 4\% on average for 3 resource types.
In the on-line setting, results showed that our proposed algorithm outperformed a greedy policy by an average improvement of 16\% but was outperformed by an Earliest Finish Time (EFT) policy by a decrease of 10\% on average.

Before continuing, we present in Section~\ref{section:relatedwork} the work related to our setting and, finally, we conclude in Section~\ref{sec:Conclusions}.

\section{Related work}
\label{section:relatedwork}

Most papers of the huge existing literature about GPUs concern specific applications.
There are only few papers dealing with generic scheduling in mixed CPU/GPU architectures, and very few of them consider precedence constraints.

From a theoretical perspective, the problem of scheduling tasks on two types of resources is more complex than the problem on parallel identical machines, $P \mid prec \mid C_{\max}$, but it is easier than the problem on unrelated machines, $R \mid prec \mid C_{\max}$.
Moreover, if all tasks are accelerated by the same factor in the GPU side, then $(CPU,GPU) \mid prec \mid C_{\max}$ coincides with the problem of scheduling on uniformly-related parallel machines, $Q \mid prec \mid C_{\max}$.
In this sense, we can say that the former is more general than the latter one; however, in our problem all tasks have only two different processing times, that makes it simpler.

For $P \mid prec \mid C_{\max}$, Graham's List Scheduling algorithm~\cite{Graham69} is a 2-approximation, while no algorithm can have a better approximation ratio assuming a particular variant of the Unique Games Conjecture~\cite{Svensson}.
Chudak and Shmoys~\cite{Chudak:1997:AAP} developed a polynomial-time $O(\log m)$-approximation algorithm for $Q \mid prec \mid C_{\max}$ while Chekuri and Bender~\cite{Chekuri:2001:EAA} proposed a faster polynomial-time approximation algorithm with the same order of worst-case performance.
To our best knowledge, no generic approximation algorithm exists for $R \mid prec \mid C_{max}$.
Polylogarithmic algorithms are proposed by Shmoys et al.~\cite{ShmoysChain} for the special case of $R \mid chain \mid C_{max}$ and by Kumar et al.~\cite{KumarForest} for $R \mid forest \mid C_{max}$.

For hybrid architectures, a 6-approximation algorithm was proposed by Kedad-Sidhoum et al.~\cite{LP6}.
In the case of independent tasks there is a $(\frac{4}{3}+\frac{1}{3k})$-approximation algorithm~\cite{DP43}.
If the tasks arrive in an on-line order, a 4-competitive algorithm was presented by Chen et al.~\cite{Guochuan} for hybrid architectures without precedence relations.

A closely related problem, in which the architecture consists of $Q \geq 2$ different types of resources and each task can be executed only on some of them, has also been studied in the literature.
This problem generalizes the dedicated processors case if each processor consists of several identical cores, while a $(Q+1)$-approximation algorithm has been proposed for it~\cite{restricted1+R}.
Note that given an allocation, the problem of scheduling in hybrid machines reduces to the above generalized dedicated processors problem.

On a more practical side, there exist some work about off-line scheduling, such as the well-known algorithm HEFT introduced by Topcuoglu et al.~\cite{HEFT}, which has been implemented on the run-time system starPU~\cite{Augonnet:StarPU}.
Another work studied the systematic comparison of various heuristics~\cite{Braun:2001:CES}.
Specifically, the authors examined 11 different heuristics.
This study provided a good basis for comparison and insights on circumstances why a technique outperforms another.
Finally, Bleuse et al.~\cite{DP43} compared their proposed $(\frac{4}{3}+\frac{1}{3k})$-approximation algorithm with HEFT.
Note that the later two approaches considered only independent tasks.

\section{Preliminaries}\label{sec:LowerBounds}

In this section we briefly present the two basic existing approaches for scheduling on heterogeneous/hybrid platforms and we discuss their theoretical efficiency by presenting lower bounds on their performance.

The first approach is the scheduling-oriented algorithm HEFT~\cite{HEFT}.
According to HEFT, the tasks are initially prioritized with respect to their precedence relations and their average processing times. Then, following this priority, tasks are scheduled with possible backfilling on the available pair of processor and time interval in which they feasibly complete as early as possible.
Note that HEFT is a heuristic that works for platforms with several heterogeneous resources and also takes into account possible communication costs.
However, even for the simpler setting without communication costs, with only two types of resources and $k=1$, HEFT cannot have a worst-case approximation guarantee better than $\frac{m}{2}$~\cite{DP43}.
This result depends only on the number of CPUs, since the example provided uses just one GPU.
The following theorem slightly improves the above result for the case of a single GPU.
More interestingly, it expresses the lower bound to the approximation ratio of HEFT using both the number of CPUs and of GPUs.

\begin{theorem}
\label{thm:lowerheft}
For any $k \leq \sqrt{m}$, the worst-case approximation ratio for HEFT is at least $\frac{m+k}{k^2}\left(1-\frac{1}{e^k}\right)$, even in the hybrid model with independent tasks.
\end{theorem}
\begin{proof}
We describe an instance that consists of independent tasks, and hence no communication costs are defined.
We also consider the hybrid platform model where we only have a set of $m$ identical CPUs and a set of $k$ identical GPUs.
Then, the rank of each task $T_j \in \mathcal{T}$ computed by HEFT is simplified as follows
\begin{equation*}
rank(T_j) = \frac{m\overline{p_j} + k\underline{p_j}}{m+k}
\end{equation*}
HEFT considers the tasks in non-increasing order with respect to their rank and assigns each task to the CPU or GPU where its completion time is minimized.
In case of ties, we assume, without loss of generality, that HEFT prefers to assign the task to a GPU, while it chooses arbitrarily between CPUs or GPUs.
Notice that, since all tasks are independent, no idle times are introduced in the schedule.

Our instance consists of $2m$ sets of $km+m^2$ tasks in total, as shown in Table~\ref{tab:lowerHEFT}.
\begin{table}
\centering
\begin{tabular}{|c|c|c|c|}
    \toprule
    Sets of tasks & \# tasks per set & $\overline{p_j}$ & $\underline{p_j}$ \\
    \midrule
    $A_i$, $1 \leq i \leq m$ & $k$ & $\left(\frac{m}{m+k}\right)^i$ & $\left(\frac{m}{m+k}\right)^i$ \\
    \midrule
    $B_i$, $1 \leq i \leq m$ & $m$ & $\left(\frac{m}{m+k}\right)^i$ & $\frac{k}{m^2}\left(\frac{m}{m+k}\right)^m$ \\
    \bottomrule
\end{tabular}
\caption{Sets of tasks composing the instance for which HEFT achieves an approximation ratio of $\frac{m+k}{k^2}\left(1-\frac{1}{e^k}\right)$.}\label{tab:lowerHEFT}
\end{table}

The rank of each task $T_j \in A_i$, $1 \leq i \leq m$, is
\begin{equation*}
rank(T_j) = \frac{(m+k)\left(\frac{m}{m+k}\right)^i}{m+k}
\end{equation*}
while the rank of each task $T_j \in B_i$, $1 \leq i \leq m$, is
\begin{equation*}
rank(T_j) = \frac{m \left(\frac{m}{m+k}\right)^i + \frac{k^2}{m^2} \left(\frac{m}{m+k}\right)^m }{m+k}
\end{equation*}

According to the above ranks, HEFT will schedule all tasks in $A_{i+1}$ (resp. $B_{i+1}$) after all tasks in $A_i$ (resp. $B_i$), $1 \leq i \leq m-1$.
Moreover, for any $T_j \in A_i$ and $T_{j'} \in B_i$, $1 \leq i \leq m$, we have
\begin{eqnarray*}
&& (m+k) \left(rank(T_j) - rank(T_{j'})\right) \\
&& = (m+k)\left(\frac{m}{m+k}\right)^i - m \left(\frac{m}{m+k}\right)^i - \frac{k^2}{m^2} \left(\frac{m}{m+k}\right)^m\\
&& = k \left( \left(\frac{m}{m+k}\right)^i - \frac{k}{m^2}\left(\frac{m}{m+k}\right)^m \right)\\
&& \geq k \left( \left(\frac{m}{m+k}\right)^m - \frac{k}{m^2} \left(\frac{m}{m+k}\right)^m \right) > 0
\end{eqnarray*}
where the last inequality holds since $k \leq m$.
For any $T_j \in B_i$ and $T_{j'} \in A_{i+1}$, $1 \leq i \leq m-1$, we have
\begin{eqnarray*}
&& (m+k) \left(rank(T_j) - rank(T_{j'})\right) \\
&& = m \left(\frac{m}{m+k}\right)^i + \frac{k^2}{m^2} \left(\frac{m}{m+k}\right)^m - (m+k)\left(\frac{m}{m+k}\right)^{i+1}\\
&& > \left(\frac{m}{m+k}\right)^i \left(m - (m+k)\frac{m}{m+k}\right) = 0
\end{eqnarray*}
Based on the above, HEFT will consider the sets of tasks according to the following order
\begin{equation*}
A_1 \prec B_1 \prec A_2 \prec B_2 \prec \cdots \prec A_i \prec B_i \prec A_{i+1} \prec \cdots \prec A_m \prec B_m 
\end{equation*}

\begin{figure}
\centering
\begin{tikzpicture}[scale=1, every node/.style={scale=.85}]
    \draw [->] (0,-1) -- (5.5,-1) node[below,font=\small] {time};
    \draw [dashed] (0,-1.05) -- (0,-.5);
    \draw [dashed] (0,2.5) -- (0,3);
    \draw [dashed] (4.9,-1.05) -- (4.9,-.5);
    \draw [dashed] (4.9,2.5) -- (4.9,3);

    \draw (-0.8, 1) node[scale=1.1] {GPUs};
    \draw [<-](-.8, -.45) -- (-.8, .7);
    \draw [->](-.8, 1.3) -- (-.8, 2.45);
    \draw (0,-.5) rectangle (4.9, 2.5);
    \draw (1,-.5) -- (1, 2.5);
    \draw (1.9, -.5) -- (1.9, 2.5);
    \draw (4.5, -.5) -- (4.5, 2.5);
    \draw (0, 1.5) -- (4.9, 1.5);
    \draw (0, .5) -- (4.9, .5);
    \draw[font=\small] (0.5, 0) node{$A_1$};
    \draw[font=\small] (1.45, 0) node{$A_2$};
    \draw[font=\small] (4.7, 0) node{$A_m$};
    \draw[font=\small] (0.5, 2) node{$A_1$};
    \draw[font=\small] (1.45, 2) node{$A_2$};
    \draw[font=\small] (4.7, 2) node{$A_m$};
    \draw (3.2, 0) node[scale=.9]{$\cdots$};
    \draw (3.2, 2) node[scale=.9]{$\cdots$};
    \draw (3.2, 1) node[scale=.9]{$\ddots$};
    \draw (.5, 1) node[scale=.9]{$\vdots$};
    \draw (1.45, 1) node[scale=.9]{$\vdots$};
    \draw (4.7, 1) node[scale=.9]{$\vdots$};

    \draw (-0.8, 4.5) node[scale=1.1]{CPUs};
    \draw [<-](-.8, 3.05) -- (-.8, 4.2);
    \draw [->](-.8, 4.8) -- (-.8, 5.95);
    \draw (0,3) rectangle (4.9,6);
    \draw (1,3) -- (1, 6);
    \draw (1.9, 3) -- (1.9, 6);
    \draw (4.5, 3) -- (4.5, 6);
    \draw (0, 4) rectangle (4.9, 5);
    
    \draw[font=\small] (.5, 3.5) node{$B_1$};
    \draw[font=\small] (1.45, 3.5) node{$B_2$};
    \draw[font=\small] (4.7, 3.5) node{$B_m$};
    \draw (3.2, 3.5) node[scale=.9]{$\cdots$};
    \draw[font=\small] (.5, 5.5) node{$B_1$};
    \draw[font=\small] (1.45, 5.5) node{$B_2$};
    \draw[font=\small] (4.7, 5.5) node{$B_m$};
    \draw (3.2, 5.5) node[scale=.9]{$\cdots$};
    \draw (.5, 4.5) node[scale=.9]{$\vdots$};
    \draw (1.45, 4.5) node[scale=.9]{$\vdots$};
    \draw (3.2, 4.5) node[scale=.9]{$\ddots$};
    \draw (4.7, 4.5) node[scale=.9]{$\vdots$};

    \draw [->] (6.3,-1) -- (8.4,-1) node[below,font=\small] {time};
    \draw [dashed] (6.3,-1.05) -- (6.3,-.5);
    \draw [dashed] (6.3,2.5) -- (6.3,3);
    \draw [dashed] (8.1,-1) -- (8.1,-.5);
    \draw [dashed] (8.1,2.5) -- (8.1,3);

    \draw (6.3,-.5) rectangle (8.1,2.5);
    \draw (6.3, 1.5) -- (8.1, 1.5);
    \draw (6.3, 0.5) -- (8.1, 0.5);
    \draw (6.6,-.5) -- (6.6,2.5);
    \draw (6.9,-.5) -- (6.9,2.5);
    \draw (7.2,-.5) -- (7.2,2.5);
    \draw [black, fill=gray!25](7.5,-.5) rectangle (8.1,2.5);
    \draw[font=\small] (6.9, 1) node{$\bigcup B_i$};

    \draw [black,fill=gray!25](6.3,3) rectangle (8.1,5);
    \draw[font=\small] [black,fill=white](6.3,5) rectangle (8.1,6);
    \draw[font=\small] [fill=white](6.3,4) rectangle (7.7,5);
    \draw [black,fill=white](6.3,3) rectangle (7.5,4);
    \draw (7, 4.5) node[scale=.9]{$\vdots$};
    \draw (6.9, 5) -- (6.9, 6);
    \draw (7.5, 5) -- (7.5, 6);
    \draw (6.7, 3) -- (6.7, 4);
    \draw (7.1, 3) -- (7.1, 4);
    \draw[font=\small] (6.6, 5.5) node{$A_1$};
    \draw[font=\small] (7.2, 5.5) node{$A_1$};
    \draw[font=\small] (7.8, 5.5) node{$A_1$};
    \draw[font=\small] (6.5, 3.5) node{$A_m$};
    \draw[font=\small] (6.9, 3.5) node{$A_m$};
    \draw[font=\small] (7.3, 3.5) node{$A_m$};

\end{tikzpicture}
\caption{Possible schedule of HEFT (left) and optimal schedule (right). Notice that the gray area represents idle time.}
\label{fig:heft}
\end{figure}
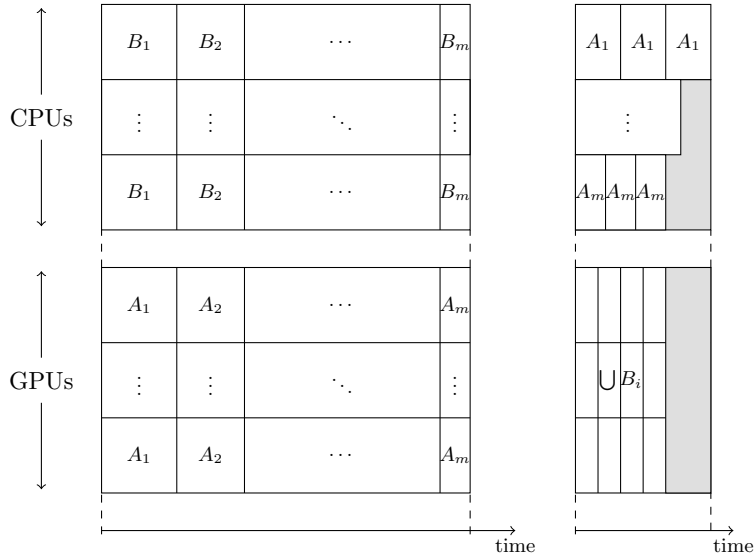

Initially, HEFT will schedule the $k$ tasks in $A_1$ in a different GPU.
Hence, to minimize the completion times of the $m$ tasks in $B_1$, each one should be scheduled on a different CPU.
Note that, all tasks in $A_1 \cup B_1$ finish at the same time, i.e., at time $\frac{m}{m+k}$.
Similarly, the tasks in $A_2$ will be scheduled on a different GPU, the tasks in $B_2$ on a different CPU, and all of them will finish at the same time, i.e., at time $\frac{m}{m+k}+\left(\frac{m}{m+k}\right)^2$.
The scheduling procedure continues in the same way for the tasks in the remaining sets.
The left-hand side of Figure~\ref{fig:heft} shows a schedule produced by HEFT.
In this schedule, all machines finish their execution at time
\begin{eqnarray*}
&& \sum_{i=1}^m\left(\frac{m}{m+k}\right)^i
= \frac{1-\left(\frac{m}{m+k}\right)^{m+1}}{1-\frac{m}{m+k}} - 1
\simeq \frac{1-\frac{m}{m+k} \frac{1}{e^k}}{1-\frac{m}{m+k}} - 1\\
&& = \frac{(m+k)e^k-m}{k e^k} - 1
= \frac{m e^k - m}{k e^k} 
= \frac{m}{k} \left(1-\frac{1}{e^k}\right)
\end{eqnarray*}

On the other hand, we can create a schedule of makespan at most $\frac{km}{m+k}$.
To see this, we assign all tasks of $A_i$, $1 \leq i \leq m$, on CPU $i$, while to each of the $k$ GPUs we assign $\frac{m^2}{k}$ different tasks of $\bigcup_{i=1}^m B_i$.
The right-hand side of Figure~\ref{fig:heft} visualizes such a schedule, whose makespan is dominated either by the load of CPU 1 or by the load of any of the GPUs.
Specifically, the makespan will be equal to
\begin{equation*}
\max \left\{ k \left(\frac{m}{m+k}\right), \frac{m^2}{k}\frac{k}{m^2}\left(\frac{m}{m+k}\right)^m \right\}
\leq \frac{km}{m+k}
\end{equation*}
Since an optimal schedule could have an even smaller makespan the theorem follows.
\end{proof}


The second approach is proposed by Kedad-Sidhoum et al.~\cite{LP6} and distinguishes the allocation and the scheduling decisions.
For the allocation phase, an integer linear program is proposed which decides the allocation of tasks to the CPU or GPU side by optimizing the standard lower bounds for the makespan of a schedule which are proposed by Graham~\cite{Graham69}, namely the critical path and the load.
To present this integer linear program, let $x_j$ be a binary variable which is equal to 1 if a task $T_j$ is assigned to the CPU side, and zero otherwise.
Let also $C_j$ be a variable that indicates the completion time of $T_j$ and $\lambda$ the variable that corresponds to the maximum over all lower bounds used., i.e., to a lower bound of the makespan.
Then, the Heterogeneous Linear Program (HLP) is as follows:
\begin{align}
& \text{minimize } \lambda \notag\\
& C_i + \overline{p_j} x_j + \underline{p_j} (1-x_j) \leq C_j && \forall T_j \in \mathcal{T}, T_i \in \Gamma^{-}(T_j) \label{ci1}\\
& \overline{p_j} x_j + \underline{p_j} (1-x_j) \leq C_j && \forall T_j \in \mathcal{T}: \Gamma^{-}(T_j)=\emptyset \label{ci6}\\
& C_j \leq \lambda && \forall T_j \in \mathcal{T} \label{ci2}\\
& \frac{1}{m} \sum_{T_j \in \mathcal{T}} \overline{p_j} x_j \leq \lambda \label{ci3}\\
& \frac{1}{k} \sum_{T_j \in \mathcal{T}} \underline{p_j} (1-x_j) \leq \lambda \label{ci4}\\
& x_j \in \{0,1\} && \forall T_j \in \mathcal{T} \label{ci5}\\
& C_j \geq 0 && \forall T_j \in \mathcal{T} \notag
\end{align}
Constraints~(\ref{ci1}),~(\ref{ci6}) and~(\ref{ci2}) describe the critical path, while Constraints~(\ref{ci3}) and (\ref{ci4}) impose that the makespan cannot be smaller than the average load on CPU and GPU sides.
Note that the particular problem of deciding the allocation to minimize the maximum over the three lower bounds is NP-hard, since it is a generalization of the PARTITION problem to which reduces if all tasks are independent, $m=k$, and $\overline{p_j}=\underline{p_j}$ for each $T_j$.

After relaxing the integrity Constraint~(\ref{ci5}), a fractional allocation can be found in polynomial time.
To get an integral solution, the variables $x_j$ are rounded as follows: if $x_j \geq \frac{1}{2}$ then $T_j$ is assigned to the CPU side, otherwise $T_j$ is assigned to the GPU side.

Finally, the Earliest Starting Time (EST) policy is applied for scheduling the tasks: at each step, the ready task with the earliest possible starting time is scheduled respecting the precedence relations and the decided allocation.
We call this algorithm HLP-EST.\\

HLP-EST achieves an approximation ratio of~6~\cite{LP6}.
The following theorem shows that this ratio is tight.

\begin{theorem}
\label{thm:lower6}
There is an instance for which HLP-EST achieves an approximation ratio of $6-O(\frac{1}{m})$.
Hence, the ratio for HLP-EST is tight.
\end{theorem}

\begin{proof}
Consider a hybrid system with an equal number of CPUs and GPUs, i.e, $m=k$.
The instance consists of $2m+3$ tasks that are partitioned into $3$ sets as shown in Table~\ref{tab:lowerHLP6}.

\begin{table}
\centering
\begin{tabular}{|c|c|c|c|}
    \toprule
    Sets of tasks & \# tasks per set & $\overline{p_j}$ & $\underline{p_j}$ \\
    \midrule
    $A$ & $1$ & $\frac{m(2m+1)}{m-1}$ & $\infty$ \\
    \midrule
    $B_1$ & $2m+1$ & $2m-1$ & $1$ \\
    \midrule
    $B_2$ & $2m+1$ & $1$ & $2m-1$ \\
    \bottomrule
\end{tabular}
\caption{Sets of tasks composing the instance for which HLP-EST achieves an approximation ratio of $6-O(\frac{1}{m})$.}\label{tab:lowerHLP6}
\end{table}

The only precedence relations exist between tasks of $B_1$ and $B_2$.
Specifically, for each task $T_j \in B_2$ we have that $\Gamma^-(T_j)=B_1$, that is no task in $B_2$ can be executed before the completion of all tasks in $B_1$.
Note that there are no precedences between tasks of the same set.

Any optimal solution of the relaxed HLP for the above instance will assign the task $T_A$ on a CPU, i.e., $x_A=1$.
Hence, the objective value of any optimal solution will be at least $\frac{m(2m+1)}{m-1}$ due to Constraints~(\ref{ci6}) and~(\ref{ci2}).
The following technical proposition shows that an optimal solution for the relaxed HLP has exactly this objective value, by describing a feasible fractional assignment for the remaining tasks.

\begin{proposition}
\label{prop:lowerlp}
There is a small constant $\epsilon>0$ for which the assignment $x_A=1$, $x_j=\frac{1}{2}$ for each $T_j \in B_1$, $x_j=\frac{1}{2}-\epsilon$ for each $T_j \in B_i$, and $\lambda=\frac{m(2m+1)}{m-1}$ corresponds to a feasible solution for the relaxed HLP.
\end{proposition}

\begin{proof}
We will show that every constraint of the relaxed HLP is satisfied by the assignment of the binary variables $x_j$ proposed in the statement and by setting $\lambda = \frac{m(2m+1)}{m-1}$.
But before this, we need to feasibly define $C_j$, for each $T_j \in \mathcal{T}$, based on Constraints~(\ref{ci1}) and~(\ref{ci6}).\\

For the task $T_A$, we set 
\begin{equation*}
C_A=\frac{m(2m+1)}{m-1}
\end{equation*}
for each task $T_j \in B_1$, we set
\begin{equation*}
C_j = \frac{1}{2}(2m-1)+\frac{1}{2} = m
\end{equation*}
while for each task $T_j \in B_2$, we set
\begin{equation*}
C_j = m + (\frac{1}{2}-\epsilon)+(\frac{1}{2}+\epsilon)(2m-1) = 2m + 2\epsilon(m-1)
\end{equation*}
satisfying by definition Constraints~(\ref{ci1}) and~(\ref{ci6}).\\

To show the feasibility of Constraint~(\ref{ci2}), it suffices to prove it for $T_A$ as well as for a task $T_j \in B_2$.
For these cases, we have
\begin{align*}
 C_A &= \frac{m(2m+1)}{m-1} = \lambda \\
 C_j &= 2m + 2\epsilon(m-1) \leq \lambda
\end{align*}
where the last inequality holds for arbitrarily small $\epsilon$, and hence Constraint~(\ref{ci2}) is satisfied.\\

For Constraint~(\ref{ci3}), we have
\begin{align*}
\sum_{T_j \in \mathcal{T}} \overline{p_j}x_j 
 & = \overline{p_A}x_A + \sum_{T_j \in B_1 \cup B_2} \overline{p_j}x_j \\
 & = \frac{m(2m+1)}{m-1} + (2m+1) \frac{2m-1}{2} + (2m+1) (\frac{1}{2}-\epsilon) \\
 & < \frac{m(2m+1)}{m-1} + m(2m+1) \\
 & = m \frac{m(2m+1)}{m-1} = m \lambda
\end{align*}
and hence it is satisfied.\\

For Constraint~(\ref{ci4}), we have
\begin{align*}
\sum_{T_j \in \mathcal{T}} \underline{p_j}(1-x_j)
 & = \underline{p_A}(1-x_A) + \sum_{T_j \in B_1 \cup B_2} \underline{p_j}(1-x_j) \\
 & = 0 + (2m+1) \frac{1}{2} + (2m+1)(2m-1) (\frac{1}{2}+\epsilon) \\
 & < m(2m+1) + \epsilon(4m^2-1) \\
 & \leq m \lambda = k \lambda
\end{align*}
where the last inequality is true for an arbitrarily small $\epsilon$, and hence the constraint is satisfied.\\

Concluding, all constraints are satisfied with $\lambda=\frac{m(2m+1)}{m-1}$, and thus the proposition holds.
\end{proof}

Given the optimal fractional assignment proposed above, HLP-EST will round the fractional variables and allocate the tasks as follows: the task $T_A$ is assigned to the CPU side, each task $T_j \in B_1$ is assigned to the CPU side, and each task $T_j \in B_2$ is assigned to the GPU side.
Then, HLP-EST schedules the tasks according to the EST policy.
However, we will argue here for any scheduling policy and thus any possible schedule.

Assuming that an algorithm has scheduled the task $T_A$ on any CPU during any interval $[t,t+\overline{p_A})$ and $m \geq 3$, there is only one meaningful family of schedules for the tasks in $B_1 \cup B_2$.
Specifically, the $2m+1$ tasks of $B_1$ will be scheduled during the interval $[0,3(2m-1))$ on the $m$ CPUs, while at least one of them completes at time $3(2m-1)$.
Then, the $2m+1$ tasks of $B_2$ will be scheduled during the interval $[3(2m-1),6(2m-1))$ on the $k=m$ GPUs, while at least one of them completes at time $6(2m-1)$.
Clearly, we should define $t$ such that $t+\overline{p_A} \leq 6(2m-1)$.
An illustration of the above schedule is given in Figure~\ref{fig:makespanJi}.

\begin{figure}
\centering
\begin{tikzpicture}[scale=1.2, every node/.style={scale=1.1}]
    \draw [font=\small] (-0.5, 2.25) node {GPUs};
    \draw [black,fill=gray!25] (0,1) rectangle (7,3.5);
    \draw (0,1.5) -- (7,1.5);
    \draw (0,2.5) -- (7,2.5);
    \draw (0,3) -- (7,3);
    \draw [font=\small,fill=white] (3*7/6,1) rectangle node {$B_2$} (4*7/6,1.5);
    \draw [font=\small,fill=white] (4*7/6,1) rectangle node {$B_2$} (5*7/6,1.5);
    \draw [font=\small,fill=white] (5*7/6,1) rectangle node {$B_2$} (6*7/6,1.5);
    \draw [font=\small,fill=white] (3*7/6,1.5) rectangle node {$\dots$} (5*7/6,2.5);
    \draw [font=\small,fill=white] (3*7/6,2.5) rectangle node {$B_2$} (4*7/6,3);
    \draw [font=\small,fill=white] (4*7/6,2.5) rectangle node {$B_2$} (5*7/6,3);
    \draw [font=\small,fill=white] (3*7/6,3) rectangle node {$B_2$} (4*7/6,3.5);
    \draw [font=\small,fill=white] (4*7/6,3) rectangle node {$B_2$} (5*7/6,3.5);

    \draw [font=\small] (-0.5, 5.25) node {CPUs};
    \draw [black,fill=gray!25] (0,4) rectangle (7,6.5);
    \draw (0,4.5) -- (7,4.5);
    \draw (0,5.5) -- (7,5.5);
    \draw (0,6) -- (7,6);
    \draw [font=\small,fill=white] (0*7/6,4) rectangle node {$B_1$} (1*7/6,4.5);
    \draw [font=\small,fill=white] (1*7/6,4) rectangle node {$B_1$} (2*7/6,4.5);
    \draw [font=\small,fill=white] (2*7/6,4) rectangle node {$B_1$} (3*7/6,4.5);
    \draw [font=\small,fill=white] (0*7/6,4.5) rectangle node {$\dots$} (2*7/6,5.5);
    \draw [font=\small,fill=white] (0*7/6,5.5) rectangle node {$B_1$} (1*7/6,6);
    \draw [font=\small,fill=white] (1*7/6,5.5) rectangle node {$A$} (2*7/6+0.1,6);
    \draw [font=\small,fill=white] (0*7/6,6) rectangle node {$B_1$} (1*7/6,6.5);
    \draw [font=\small,fill=white] (1*7/6,6) rectangle node {$B_1$} (2*7/6,6.5);
    \draw [font=\small,fill=white] (2*7/6,6) rectangle node {$B_1$} (3*7/6,6.5);

    \draw [thick,->](0,0.5) -- (7.5,0.5);
    \draw [dotted] (0*7/6,0.4) -- (0*7/6,6.5);
    \draw [dotted] (1*7/6,0.0) -- (1*7/6,6.5);
    \draw [dotted] (2*7/6,0.4) -- (2*7/6,6.5);
    \draw [dotted] (3*7/6,0.0) -- (3*7/6,6.5);
    \draw [dotted] (4*7/6,0.4) -- (4*7/6,6.5);
    \draw [dotted] (5*7/6,0.0) -- (5*7/6,6.5);
    \draw [dotted] (6*7/6,0.4) -- (6*7/6,6.5);
    \node [font=\small] at (0*7/6,0.2) {$0$};
    \node [font=\small] at (1*7/6,-0.2) {$(2m-1)$};
    \node [font=\small] at (2*7/6,0.2) {$2(2m-1)$};
    \node [font=\small] at (3*7/6,-0.2) {$3(2m-1)$};
    \node [font=\small] at (4*7/6,0.2) {$4(2m-1)$};
    \node [font=\small] at (5*7/6,-0.2) {$5(2m-1)$};
    \node [font=\small] at (6*7/6,0.2) {$6(2m-1)$};
\end{tikzpicture}
\caption{Resulting schedule of HLP-EST for the proposed instance. Notice that the gray areas represent idle times.}
\label{fig:makespanJi}
\end{figure}
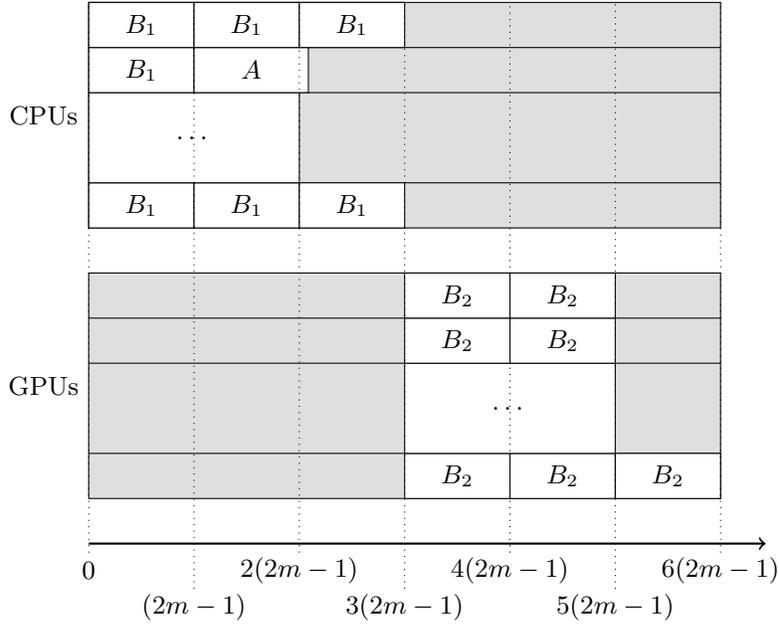

The makespan of the created schedule is equal to $6(2m-1)$, while Proposition~\ref{prop:lowerlp} implies a feasible solution for the relaxed HLP of objective value $\frac{m(2m+1)}{m-1}$.
Hence, the approximation ratio achieved for this instance is
\begin{equation*}
\frac{6(2m-1)}{\frac{m(2m+1)}{m-1}} = 6 - O\left(\frac{1}{m}\right)
\end{equation*}
and the theorem follows.
\end{proof}

Note that the proof of the previous theorem implies a stronger result since the worst case example does not depend on which scheduling policy will be applied after the allocation step, and hence the following corollary holds.

\begin{corollary} \label{cor:hlp}
Any scheduling policy which is applied after the allocation decisions taken by the rounding of an optimal solution of the relaxed HLP leads to an approximation algorithm of ratio at least $6-O(\frac{1}{m})$.
\end{corollary}

\section{Algorithms}\label{sec:Algorithms}

In this section we propose algorithms for both the off-line and on-line settings of the addressed problem.

\subsection{Off-line setting}\label{ssec:offlineSetting}

We propose in the following a new scheduling policy which prioritizes the tasks based on the solution obtained for HLP after the rounding step.

The motivation of assigning priorities to the tasks is for taking into account the precedence relations between them.
More specifically, we want to prioritize the scheduling of \emph{critical tasks}, i.e., the tasks on the critical path, before the remaining (less critical) tasks.

To do this, we define for each task $T_j$ a rank $Rank(T_j)$ in the same sense as in HEFT.
However, in our case, the rank of each task depends on the allocation given by HLP, while in HEFT it depends on the average processing time of the task.
Specifically, the rank of each task $T_j$ is computed after the rounding operation of the assignment variable $x_j$ and corresponds to the length, in the sense of processing time, of the longest path between this task and its last descendant in the precedence graph.
Thus, each task will have a larger rank than all its descendants.
The rank of the task $T_j$ is recursively defined as follows:
\[ Rank(T_j) = \overline{p_j}x_j + \underline{p_j}(1-x_j) + \max_{i \in \Gamma^+(T_j)}\{Rank(T_i)\} \]
After ordering the tasks in non-increasing order with respect to their ranks, we apply the standard List Scheduling algorithm adapted to two types of resources and taking into account the rounding of the assignment variables $x_j$.
We call the above described policy Ordered List Scheduling (OLS), while the newly defined algorithm (including the allocation) is denoted by HLP-OLS.\\

Although this policy performs well in practice, as we will see in Section~\ref{sec:Experiments}, its approximation ratio cannot be better than 6 due to the lower bound presented in Theorem~\ref{thm:lower6}.
On the other hand, it is quite easy to see that HLP-EST and HLP-OLS have the same approximation ratio by following the same reasoning as in Lemmas 4 and 5 of Kedad-Sidhoum et al.~\cite{LP6}.

Consider a schedule produced by HLP-OLS and partition the time interval $I = [0,C_{max})$ into two subsets $I_{CP}$ and $I_{W}$. The set $I_{CP}$ contains every time slot where at least one processor of each type is idle, while the set $I_W$ consists of the remaining time slots in I, i.e., $I_W = I \backslash I_{CP}$. We then can divide the set $I_W$ into two possibly non-disjoint subsets $I_{CPU}$ (resp. $I_{GPU}$) containing the time slots where all the CPUs (resp. GPUs) are busy.

Denoting by $|I|$ the number of unitary time slots in an interval $I$ we have
\begin{align*}
    |I_{CP}| &\leq CP \\
    |I_{CPU}| &\leq \frac{W_{CPU}}{m} \\
    |I_{GPU}| &\leq \frac{W_{GPU}}{k}
\end{align*}
where $CP$, $W_{CPU}$ and $W_{GPU}$ denote respectively the length of the critical path, the total work load on all CPUs and the total work load on all GPUs of the considered schedule. 
With $C_{max}^*$ being the makespan of the optimal schedule, each of the 3 above inequalities are bounded by $2C_{max}^*$.
Thus, we can deduce the following bound for the makespan of HLP-OLS
\begin{align*}
    C_{max} &\leq |I_{CPU}| + |I_{GPU}| + |I_{CP}| \\
    &\leq \frac{W_{CPU}}{m} + \frac{W_{GPU}}{k} + CP \\
    &\leq 6C_{max}^*
\end{align*}

\begin{corollary}
HLP-OLS achieves an approximation ratio of 6. This ratio is tight.
\end{corollary}

\subsection{On-line setting}
\label{ssec:onlineSetting}

In the HLP-EST algorithm, an integer linear program is used to find an efficient allocation of each task to the CPU or GPU side.
Although this program optimizes the classical lower bounds for the makespan, and hence informally optimizes the allocation, the resolution of its relaxation has a high complexity in practice and cannot be used in an on-line setting.
For this reason, we would like to explore some greedy, low complexity, policies.
In this direction, we initially propose the following three simple greedy rules:
\begin{description}
\item[R1] If $\frac{\overline{p_j}}{m} \leq \frac{\underline{p_j}}{k}$ then assign $T_j$ to the CPU side, else assign it to the GPU side.
\item[R2] If $\frac{\overline{p_j}}{\sqrt{m}} \leq \frac{\underline{p_j}}{\sqrt{k}}$ then assign $T_j$ to the CPU side, else assign it to the GPU side.
\item[R3] If $\overline{p_j} \leq \underline{p_j}$ then assign $T_j$ to the CPU side, else assign it to the GPU side.
\end{description}
However, these rules do not take into account neither the critical path nor the actual schedule and they cannot guarantee a bounded approximation ratio.\\

In what follows, we propose to use a more enhanced set of rules which combines a rule based on the structure of the actual schedule with R2, in a similar way as in the 4-competitive algorithm proposed by Chen et al.~\cite{Guochuan} for the on-line problem with independent tasks.

To describe the new rule, we define $\tau_{gpu}$ to be the earliest time when at least one GPU is idle.
Let also $R_{j, gpu} = \max\{\tau_{gpu}, \max_{T_i\in \Gamma^{-}(T_j)}\{C_i\}\}$ be the \emph{ready time} of $T_j$ for GPUs, i.e., the earliest time at which $T_j$ can be executed on a GPU.
Then, the new enhanced set of rules is defined as follows:
\begin{description}
\item[Step 1:] If $\overline{p_j} \geq R_{j, gpu} + \underline{p_j}$ then assign $T_j$ to the GPU side.
\item[Step 2:] Otherwise apply R2.
\end{description}

This set of rules can be combined with a greedy List Scheduling policy that schedules each task as early as possible on the CPU or GPU side already decided by the rules.
We call the algorithm obtained by this combination ER-LS (Enhanced Rules - List Scheduling).
In the following, we give upper and lower bounds for the competitive ratio of ER-LS.
\medskip

\begin{theorem}
\label{theorem:lowerON}
ER-LS is at most a $(4\sqrt{\frac{m}{k}})$-competitive algorithm.
\end{theorem}

\begin{proof}
Let ${W_{CPU}}$, ${W_{GPU}}$ and $CP$ be the total load on all CPUs, the total load on all GPUs and the length of the critical path of a schedule produced by the algorithm, respectively.

With the same reasoning than in Section~\ref{ssec:offlineSetting} and given a schedule produced by ER-LS we can observe that
\begin{equation}
    C_{max} \leq \frac{W_{CPU}}{m} + \frac{W_{GPU}}{k} + CP \label{eqn0}
\end{equation}
In the following, we bound the average load of both sides~($\frac{W_{CPU}}{m} + \frac{W_{GPU}}{k}$) by $3\sqrt{\frac{m}{k}}C_{max}^*$ and the length of the critical path by $\sqrt{\frac{m}{k}}C_{max}^*$.
Recall that $C_{max}^*$ denotes the makespan of the optimal off-line solution of the instance.\\

We denote by $\mathcal{SA}_{cpu}$ (resp. $\mathcal{SA}_{gpu}$) the set containing the tasks placed on the CPU (resp. GPU) side in both a solution of the algorithm and the optimal solution, by $\mathcal{SB}_{gpu}$ the set containing tasks placed by Step 1 on the GPU side in a solution of the algorithm but on the CPU side in the optimal solution, and by $\mathcal{SC}_{cpu}$ (resp. $\mathcal{SC}_{gpu}$) the set containing tasks placed by Step 2 on the CPU (resp. GPU) side in a solution of the algorithm but on the GPU (resp. CPU) side in the optimal solution.

We also denote by $sa_{cpu}$, $sa_{gpu}$, $sb_{gpu}$, $sc_{cpu}$ and $sc_{gpu}$ the sum of processing times of all tasks in the sets $\mathcal{SA}_{cpu}$, $\mathcal{SA}_{gpu}$, $\mathcal{SB}_{gpu}$, $\mathcal{SC}_{cpu}$ and $\mathcal{SC}_{gpu}$, respectively. Note that we use here the processing times according to the allocation of ER-LS\\

\noindent\textbf{Bounding the loads.}\\
Consider $T_{j_0}$ to be the last finishing task in $\mathcal{SB}_{gpu}$.
Since the task is scheduled according to Step 1, we know that $\overline{p_{j_0}} \geq R_{j_0, gpu} + \underline{p_{j_0}} \geq \frac{sb_{gpu}}{k}$. We also know that $T_{j_0}$ is scheduled on a CPU in the optimal solution so we have $\overline{p_{j_0}} \leq C_{max}^*$ and thus: $\frac{sb_{gpu}}{k} \leq C_{max}^*$.

Each task in $\mathcal{SC}_{gpu}$ is scheduled on the CPU side in the optimal solution. According to Step 2, the total processing times of tasks in $\mathcal{SC}_{gpu}$ in the optimal solution is at least $\sqrt{\frac{m}{k}}sc_{gpu}$, so we have for the cpu side $\frac{sa_{cpu} + \sqrt{\frac{m}{k}}sc_{gpu}}{m} \leq C_{max}^*$.
The same reasoning for the GPU side gives $\frac{sa_{gpu} + \sqrt{\frac{k}{m}}sc_{gpu}}{k} \leq C_{max}^*$.

By adding the three inequalities we have the following:
\[ \frac{sb_{gpu}}{k} + \frac{sa_{cpu} + \sqrt{\frac{m}{k}}sc_{gpu}}{m} + \frac{sa_{gpu} + \sqrt{\frac{k}{m}}sc_{cpu}}{k} \leq 3C_{max}^* \]

By separating the loads on CPU and on GPU on the left-hand side of the above inequality and taking into account that $m \geq k$ we have
\[ \frac{sa_{cpu}}{m}+\frac{sc_{cpu}}{\sqrt{mk}} \geq \frac{sa_{cpu} + sc_{cpu}}{m} \geq \sqrt{\frac{k}{m}}\frac{sa_{cpu} + sc_{cpu}}{m} \]
and
\[ \frac{sa_{gpu} + sb_{gpu}}{k} + \frac{sc_{gpu}}{\sqrt{mk}} \geq \frac{sa_{gpu} + sb_{gpu}}{k} + \frac{sc_{gpu}}{k}\sqrt{\frac{k}{m}} \geq \sqrt{\frac{k}{m}}\frac{sa_{gpu} + sb_{gpu} + sc_{gpu}}{k} \]
Summing these two bounds we finally obtain
\[ \sqrt{\frac{k}{m}} (\frac{sa_{cpu} + sc_{cpu}}{m} + \frac{sa_{gpu} + sb_{gpu} + sc_{gpu}}{k} ) \leq 3C_{max}^* \]
and thus
\begin{equation} 
\frac{W_{CPU}}{m} + \frac{W_{GPU}}{k} \leq 3\sqrt{\frac{m}{k}}C_{max}^* \label{eqn1}
\end{equation}

\medskip
\noindent\textbf{Bounding the critical path.}\\
Consider the sets $\mathcal{SA}_{cpu}^{CP} \subseteq \mathcal{SA}_{cpu}$, $\mathcal{SA}_{gpu}^{CP} \subseteq \mathcal{SA}_{gpu}$, $\mathcal{SB}_{gpu}^{CP} \subseteq \mathcal{SB}_{gpu}$, $\mathcal{SC}_{cpu}^{CP} \subseteq \mathcal{SC}_{cpu}$ and $\mathcal{SC}_{gpu}^{CP} \subseteq \mathcal{SC}_{gpu}$ to be the sets containing only the tasks belonging to the critical path obtained by the algorithm, with the same notation in lower case for the sum of processing times of all tasks in each set and the same notation with a star $^*$ for the sum of processing times of all tasks in the optimal solution.

For the sets $\mathcal{SA}_{cpu}^{CP}$ and $\mathcal{SA}_{gpu}^{CP}$, by definition, we have 
\begin{equation*}
sa_{cpu}^{CP} + sa_{gpu}^{CP} = sa_{cpu}^{CP^*} + sa_{gpu}^{CP^*}
\end{equation*}

According to Step 1, every task in $\mathcal{SB}_{gpu}^{CP}$ has a processing time smaller than that in the optimal solution, so $sb_{gpu}^{CP} \leq sb_{gpu}^{CP^*}$.
According to Step 2, every task $T_j$ in $\mathcal{SC}_{cpu}^{CP}$ (resp. $\mathcal{SC}_{gpu}^{CP}$) verifies $\overline{p_j} \leq \sqrt{\frac{m}{k}}\underline{p_j}$ (resp. $\underline{p_j} \leq \sqrt{\frac{k}{m}}\overline{p_j}$), so we have $sc_{cpu}^{CP} \leq \sqrt{\frac{m}{k}}sc_{cpu}^{CP^*}$ and $sc_{gpu}^{CP} \leq \sqrt{\frac{m}{k}}sc_{gpu}^{CP^*}$.

By summing the previous inequalities for the critical path we get
\begin{align*}
    CP &= sa_{cpu}^{CP} + sa_{gpu}^{CP} + sb_{gpu}^{CP} + sc_{cpu}^{CP} + sc_{gpu}^{CP} \\
    &\leq \sqrt{\frac{m}{k}}(sa_{cpu}^{CP^*} + sa_{gpu}^{CP^*} + sb_{gpu}^{CP^*} + sc_{cpu}^{CP^*} + sc_{gpu}^{CP^*}) \leq \sqrt{\frac{m}{k}}CP^*
\end{align*}
Since $CP^* \leq C_{max}^*$, we have $CP \leq \sqrt{\frac{m}{k}}C_{max}^*$ and, combining this inequality with Equations (\ref{eqn0}) and (\ref{eqn1}), the theorem follows.

\end{proof}

As the following theorem shows, the competitive ratio of ER-LS is almost tight and we cannot expect a much better analysis for its upper bound.

\begin{theorem}
There is an instance for which ER-LS achieves a competitive ratio of $\sqrt{\frac{m}{k}}$.
\end{theorem}

\begin{proof}
Consider a hybrid system with $m$ CPUs and $k\leq m$ GPUs.
The instance consists of $m+k$ tasks that are partitioned into 2 sets as shown in Table~\ref{tab:lowerERLS}.

\begin{table}
\centering
\begin{tabular}{|c|c|c|c|}
	\toprule
	Type & Number of task & Processing time on CPU/GPU \\
	\midrule
	A & $k$ & $\sqrt{m} ~/~ \sqrt{m}$ \\
	\midrule
	B & $m$ & $ \sqrt{m} ~/~ \sqrt{k}$ \\
	\bottomrule
\end{tabular}
\caption{Sets of tasks composing the instance for which ER-LS achieves a competitive ratio of $\sqrt{\frac{m}{k}}$.}\label{tab:lowerERLS}
\end{table}

The $k$ tasks of type A are independent to each other and the $m$ tasks of type B are with precedence constraints as follows:
\[ B_1 \prec B_2 \prec \cdots \prec B_m\]

The tasks are ordered in a list by first taking all tasks of type A and then the tasks of type B respecting the precedences.\\

The ER-LS algorithm will first place the $k$ tasks of type A on a GPU according to Step 1. The completion time of these tasks is then $\sqrt{m}$.
Then, since $\sqrt{m} \leq \sqrt{m} + \sqrt{k}$, the task $B_1$ will be placed on a CPU according to Step 2, with completion time $\sqrt{m}$. The task $B_2$ will also be placed on a CPU according to Step 2, starting at time $\sqrt{m}$ and completing at time $2\sqrt{m}$. With the same reasoning, each task $B_i$, $i\in \{1, m\}$ is placed on a CPU according to Step 2 starting at time $(i-1)\sqrt{m}$ and completing at time $i\sqrt{m}$.

Thus, the schedule produced by ER-LS for this instance has a makespan of $C_{max} = m\sqrt{m}$.

An optimal schedule would have all tasks of type $A$ placed on the CPU side with a completion time for each task of $\sqrt{m}$. The tasks of type $B$ would be placed on the GPU side with a completion time for each task $B_i$, $i\in \{1,m\}$, of $i\sqrt{k}$. Thus, the optimal makespan is $C_{max}^* = m\sqrt{k}$.

Hence, ER-LS achieves a competitive ratio $\frac{C_{max}}{C_{max}^*} = \sqrt{\frac{m}{k}}$ for this instance and the theorem holds.
\end{proof}

\section{Generalization on Q Resource Types}\label{sec:generalization}

We now generalize the HLP-EST algorithm to extend the addressed problem to $Q\geq2$ different types of identical processors.
Before continuing we need some additional notations.
Let $M_q$ be the set of processors of type $q$, $1 \leq q \leq Q$, and $m_q=|M_q|$ its size.
The execution of a task $T_j \in \mathcal{T}$ on a processor of type $q$, $1 \leq q \leq Q$, takes $p_{j,q}$ time units.\\

In what follows we adapt HLP to take into account more resource types.
To do this, we introduce a binary variable $x_{j,q}$ which indicates if the task $T_j \in \mathcal{T}$ is assigned to the resource type $q$.
As before, let $C_j$ be a variable corresponding to the completion time of $T_j$ and $\lambda$ be the variable that represents a lower bound to the makespan.
Then, we consider the following modification of HLP, which we call QHLP:
\begin{align}
& \text{minimize } \lambda \notag\\
& C_i + \sum_{q=1}^{Q}{p_{j,q}x_{j,q}} \leq C_j && \forall T_j \in \mathcal{T}, T_i \in \Gamma^{-}(T_j) \label{gi1}\\
& \sum_{q=1}^{Q}{p_{j,q}x_{j,q}} \leq C_j && \forall T_j \in \mathcal{T}: \Gamma^{-}(T_j)=\emptyset \label{gi2}\\
& C_j \leq \lambda && \forall T_j \in \mathcal{T} \label{gi3}\\
& \frac{1}{m_q} \sum_{T_j \in \mathcal{T}} p_{j,q}x_{j,q} \leq \lambda && 1 \leq q \leq Q \label{gi4}\\
& \sum_{q=1}^{Q}{x_{j,q}} = 1 && \forall T_j \in \mathcal{T} \label{gi5} \\
& x_{j,q} \in \{0,1\} && \forall T_j \in \mathcal{T},~1 \leq q \leq Q \label{gi6}\\
& C_j \geq 0 && \forall T_j \in \mathcal{T} \notag\\
& \lambda \geq 0 \notag
\end{align}
The main difference here concerns Constraint~(\ref{gi5}) which assures that each task is integrally assigned to exactly one type of resources.

After relaxing the integrity Constraint~(\ref{gi6}) of QHLP, we can solve in polynomial time the obtained relaxation.
To get an integral allocation, we assign each task $T_j$ to the resource type $q'$ for which the assignment variable $x_{j,q}$ have the greatest value, i.e., $q'=\text{argmax}_{1 \leq q \leq Q}\{x_{j,q}\}$.
In other words, for such $q'$ we set $x_{j,q'} = 1$ and $x_{j,q} = 0$ for any $q \neq q'$.
In case of ties, we give priority to the resource type in which $T_j$ has the smallest processing time.
Once the assignment step is done, we use the Earliest Starting Time policy taking into account the precedence constraints as well as the allocation provided by the rounding of $x_{j,q}$ variables.
We call this algorithm QHLP-EST.

Then, the following theorem holds.

\begin{theorem}
\label{thm:analyseq}
QHLP-EST achieves an approximation ratio of $Q(Q+1)$. This ratio is tight.
\end{theorem}

\begin{proof}
We analyze the structure of a schedule produced by the algorithm to give an upper bound on the approximation ratio.
The analysis of the algorithm and the structure of the schedule are similar to the ones of Kedad-Sidhoum \emph{et al.} for the HLP-EST algorithm \cite{LP6}.

We denote by $W_q$, $1 \leq q \leq Q$, the total load on all processors of type $q$ in the obtained schedule.
We also denote by $C_{max}^{R}$, $W_q^{R}$ and $L^{R}$ the objective value, the total load on all processors of type $q$ and the length of the longest path in the fractional optimal solution of the relaxed QHLP, respectively.
Finally, we define by $C_{\max}^{*}$ the optimal makespan over all feasible schedules for our problem.
Then, the following inequalities hold.
\begin{align}
 \label{ineqQ:L}L^{R} &\leq C_{\max}^{R} \leq C_{\max}^{*}\\
 \label{ineqQ:Wq}\frac{W_q^{R}}{m_q} &\leq C_{\max}^{R} \leq C_{\max}^{*} ,~1 \leq q \leq Q
\end{align}

To analyze the structure of the schedule, we partition the time interval of the schedule $I=[0, C_{max})$ into two disjoint subsets of intervals $I_{CP}$ and $I_{W}$.
The set $I_{CP}$ contains every time slot where at least one processor of each type is idle, while the set $I_W$ consists of the remaining time slots in $I$, i.e., $I_W = I \backslash I_{CP}$.
We then can divide the set $I_W$ into $Q$, possibly non-disjoint, subsets $I_q$, $1 \leq q \leq Q$, which contain respectively every time slot where all processors of type $q$ are busy.
Henceforth, we denote by $|I|$ the length of $I$, i.e. the number of unitary time slots in $I$.
Then, we have that
\[ C_{\max} = |I| \leq |I_{CP}| + \sum_{q=1}^{Q}{|I_q|} \]

In the following, we will bound above by $QC_{max}^{*}$ the length of the subset $I_{CP}$ and each subsets $I_i$, $1 \leq i \leq Q$.

Due to the rounding policy, we know that if $x_{j,q}=1$ then $x_{j,q}^{R} \geq \frac{1}{Q}$.
Hence, we have
\begin{align}
\label{ineqQ:xj} x_{j,q} \leq Q \cdot x_{j,q}^{R} ~~\forall T_j\in \mathcal{T}, 1 \leq q \leq Q
\end{align}

Consider first the subset of intervals $I_{CP}$.
There is a directed path $\mathcal{P}$ of tasks being executed during any time slot in $I_{CP}$.
The construction of $\mathcal{P}$ is the same as described by Kedad-Sidhoum \emph{et al.}~\cite{LP6}.
Since the directed path $\mathcal{P}$ covers every time slot in $I_{CP}$, the length of $I_{CP}$ is smaller than the length of $\mathcal{P}$ and the length of $\mathcal{P}$ in the optimal solution of $QHLP$, noted $\mathcal{P}^{R}$, is smaller than $L^{R}$.
Thus, using the inequalities~(\ref{ineqQ:L}) and~(\ref{ineqQ:xj}), we have the following bound:
\begin{eqnarray*}
|I_{CP}| \leq |\mathcal{P}| & \leq & \sum_{j\in \mathcal{P}}{\sum_{q=1}^{Q}{p_{j,q}x_{j,q}}} \\
& \leq & Q \sum_{j\in \mathcal{P}}{\sum_{q=1}^{Q}{p_{j,q}x_{j,q}^{R}}} = Q \cdot | \mathcal{P}^R | \\
& \leq & Q \cdot L^R \leq Q \cdot C_{\max}^*
\end{eqnarray*}

Consider now each subset $I_q$, $1 \leq q \leq Q$.
For each time slot in $I_q$ all processors of type $q$ are busy, so $|I_q|$ is smaller than the average load on all the processors of type $q$.
Using the inequalities~(\ref{ineqQ:Wq}) and~(\ref{ineqQ:xj}), we have the following bounds:
\begin{eqnarray*}
|I_q| \leq \frac{W_q}{m_q} & \leq & \frac{1}{m_q}\sum_{x_{j,q}=1}{p_{j,q}} \\
& \leq & \frac{Q}{m_q}\sum_{j\in V}{p_{j,q}x_{j,q}^{R}} \\
& \leq & Q \cdot \frac{W_q^R}{m_q} \leq Q \cdot C_{\max}^*
\end{eqnarray*}

Thus, by combining the calculated bounds we get
\begin{eqnarray*}
C_{\max} = |I| & \leq & |I_{CP}| + \sum_{q=1}^{Q}{|I_q|} \\
& \leq & Q(Q+1)C_{max}^{*}
\end{eqnarray*}

The tightness comes directly from Theorem~\ref{thm:lower6}, and hence the theorem follows.
\end{proof}

The same reasoning can be used to show that QHLP-OLS achieves an approximation ratio of $Q(Q+1)$ and that this ratio is tight.

\section{Experiments}
\label{sec:Experiments}
In this section we compare the performance of various scheduling algorithms by a simulation campaign using a benchmark composed of 6 parallel applications.
First, we compare the off-line algorithms for the studied problem with 2 and 3 resource types. We then compare the on-line algorithms for 2 resource types.

\subsection{Benchmark}
The benchmark is composed of 5 applications generated by \emph{Chameleon}, a dense linear algebra software which is part of the MORSE project~\cite{agullo2012morse}, and a more irregular application (\emph{fork-join}) generated using \emph{GGen}, a library for generating directed acyclic graphs~\cite{GGen:simutools10}.

The applications of \emph{Chameleon}, named \emph{getrf}, \emph{posv}, \emph{potrf}, \emph{potri} and \emph{potrs}, are composed of multiple sequential basic tasks of linear algebra.
Different number, denoted by $nb\_blocks$, and sizes, denoted by $block\_size$, of sub-matrices have been used for the applications; specifically we set $nb\_blocks\in \{5, 10, 20\}$ and $block\_size\in\{64, 128, 320, 512, 768, 960\}$.
Different tilings of the matrices have been used, varying the number of sub-matrices denoted by $nb\_blocks$ and the size of the sub-matrices denoted by $block\_size$. 
Table~\ref{table:nbTasks1} shows the total number of tasks for each application and each value of $nb\_blocks$. Notice that $bloc\_size$ does not impact the number of tasks.

\begin{table}
\centering
\begin{tabular}{|c|c|c|c|c|c|}
	\toprule
	\backslashbox{\bf nb$\_$blocks}{\bf app} & {\bf getrf} & {\bf posv} & {\bf potrf} & {\bf potri} & {\bf potrs}\\
	\midrule
	{\bf 5} & 55 & 65 & 35 & 105 & 30 \\
	\midrule
	{\bf 10} & 385 & 330 & 220 & 660 & 110 \\
	\midrule
	{\bf 20} & 2870 & 1960 & 1540 & 4620 & 420 \\
	\bottomrule
\end{tabular}
\caption{Total number of task of each Chameleon application as a function of $nb\_blocks$.}\label{table:nbTasks1}
\end{table}

For the setting with 2 resource types, the applications were executed with the runtime StarPU~\cite{Augonnet:StarPU} on a machine with two Dual core Xeon E7 v2 with a total of 10 physical cores with hyper-threading of 3 GHz and 256 GB of RAM. The machine had 4 GPUs NVIDIA Tesla K20 (Kepler architecture) with each 5 GB of memory and 200 GB/s of bandwidth.

For 3 resource types, the applications were executed with the runtime StarPU~\cite{Augonnet:StarPU} on an Intel  Dual core i7-5930k machine with a total of 6 physical cores with hyper-threading of 3.5 GHz and 12 GB of RAM.
This machine had 2 NVIDIA GPUs: a GeForce GTX-970 (Maxwell architecture) with 4 GB of memory and 224 GB/s of bandwidth; and a Quadro K5200 (Kepler architecture) with 8 GB of memory and 192 GB/s of bandwidth.

The running time of each task composing the applications was stored for each resource type.\\

The \emph{fork-join} application corresponds to a real situation where the execution starts sequentially and then forks to $width$ parallel tasks.
The results are aggregated by performing a join operation, completing a phase. 
This procedure can be repeated $p$ times, the number of phases. 
For our experiments, we used $p\in\{2, 5, 10\}$ and $width \in\{100, 200, 300, 400, 500\}$. Table~\ref{table:nbTasks2} shows the total number of tasks for each value of $width$ and $p$.
The processing time of each task on CPU was computed using a Gaussian distribution with center $p$ and standard deviation $\frac{p}{4}$.
We established various acceleration factors for the processing time on the first type of GPU.
In all configurations, there are five percent of parallel tasks in each phase, randomly chosen, with an acceleration factor uniformly chosen in $[0.1,0.5]$ while the remaining tasks have an acceleration factor in $[0.5,50]$. We used the same process to generate the processing times for the second type of GPU.

\begin{table}
\centering
\begin{tabular}{|c|c|c|c|c|c|}
	\toprule
	\backslashbox{\bf p}{\bf width} & {\bf 100} & {\bf 200} & {\bf 300} & {\bf 400} & {\bf 500} \\
	\midrule
	{\bf 2} & 203 & 403 & 603 & 803 & 1003 \\
	\midrule
	{\bf 5} & 506 & 1006 & 1506 & 2006 & 2506 \\
	\midrule
	{\bf 10} & 1011 & 2011 & 3011 & 4011 & 5011 \\
	\bottomrule
\end{tabular}
\caption{Total number of tasks of the \emph{fork-join} application as a function of \emph{width} and \emph{p}.}
\label{table:nbTasks2}
\end{table}

The data sets and other information are available under Creative Commons Public License\footnote{Hosted at: \url{https://github.com/marcosamaris/heterogeneous-SWF}, last visited on Nov. 2017.}.

\subsection{Off-line setting}

\noindent\textbf{Algorithms and machine configurations.}\\
We compared the performance, in terms of makespan, of HLP-OLS (Section~\ref{ssec:offlineSetting}) with HLP-EST and HEFT (Section~\ref{sec:LowerBounds}). The algorithms were implemented in Python (v. 2.7.6). The command-line \emph{glpsol} (v. 4.52) solver of the GLPK package was used for the linear program.
Each algorithm was implemented with a second version adapted for 3 types of resources, using the generalization of the algorithms presented in Section~\ref{sec:generalization} for the 2 linear program-based algorithms. We denote by QHLP-EST, QHLP-OLS and QHEFT these algorithms for 3 resource types.

For the machine configurations, we determined different sets of pairs (Nb$\_$CPUs, Nb$\_$GPUs). Specifically, we used 16, 32, 64 and 128 CPUs with 2, 4, 8 and 16 GPUs for a total of 16 machine configurations for the case with 2 resource types.
For the case with 3 resource types, we determined different sets of triplets (Nb$\_$CPUs, Nb$\_$GPU1s, Nb$\_$GPU2s) with the name numbers of CPUs and for either types of GPUs, for a total number of 64 machine configurations.

We executed the algorithms only once with each combination of application and machine configuration since all algorithms are deterministic.
For each run, we stored the optimal objective solution of the linear program, denoted by $LP^*$, and the makespans of the six algorithms.

For the biggest instances of applications, the linear program resolution took about 100 seconds while the running time of each algorithm took at most 10 seconds, once a solution of HLP was found for the linear program-based algorithms.\\

\noindent\textbf{Results for 2 resource types.}\\
To study the performance of the 3 algorithms we computed the ratio between each makespan and the optimal solution $LP^*$ of the linear program HLP, which corresponds to a good lower bound of the optimal makespan.
Fig.~\ref{fig:plotR} shows the ratio of each instance of application and configuration. Notice that the red / bigger dot represents the mean value of the ratio for each application. We can see that HLP-EST is outperformed, on average, by the two other algorithms. The performances of HLP-OLS and HEFT are quite similar, on average, but we observe that HEFT creates more outliers.

\begin{figure}
\centering
\includegraphics[width=.6\textwidth]{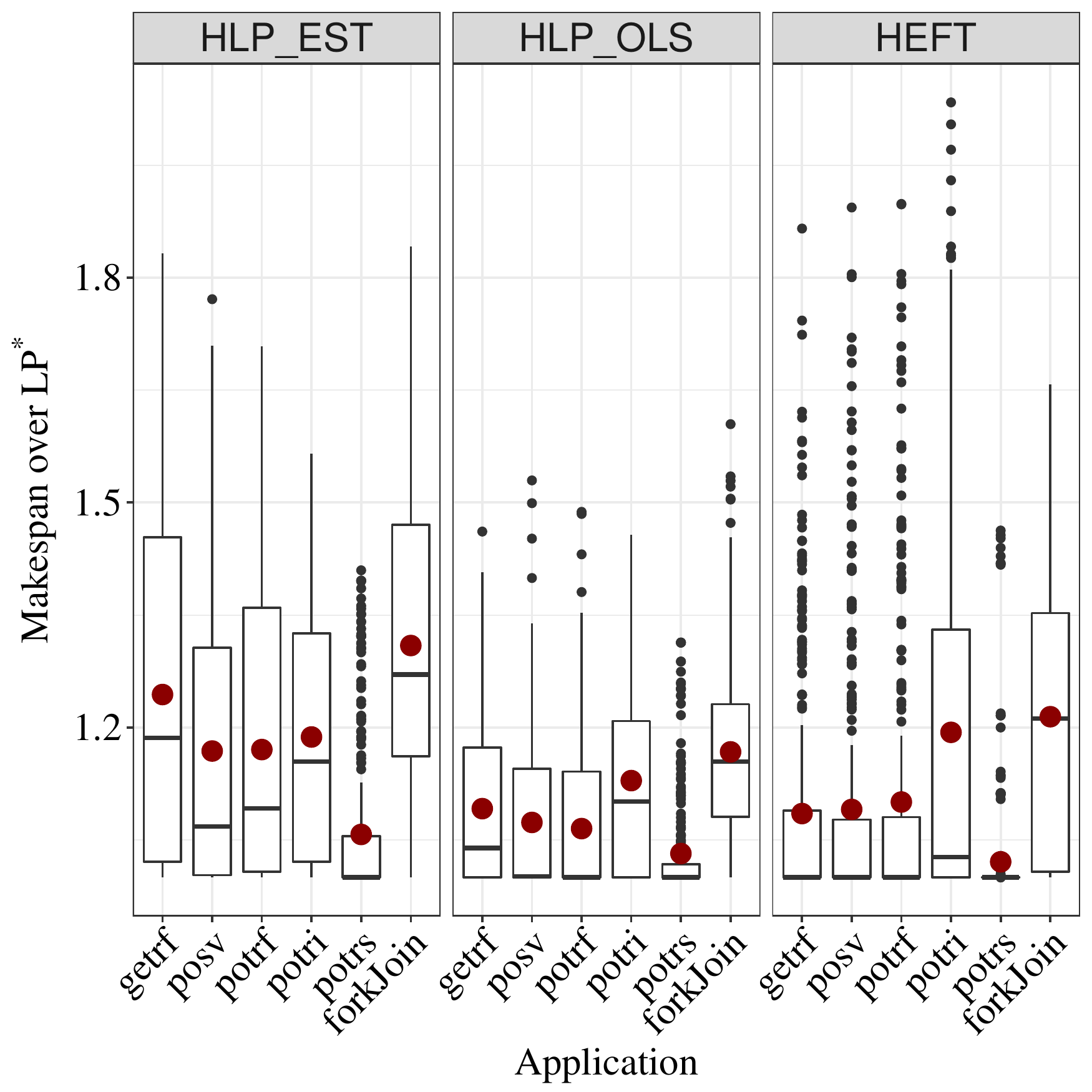}
\caption{Ratio of makespan over $LP^*$ for each instance, grouped by application for the off-line algorithms with 2 resource types.}
\label{fig:plotR}
\end{figure}

Fig.~\ref{fig:plotComp5} compares more specifically the two HLP-based algorithms (left), and the algorithms HLP-OLS and HEFT (right), by showing the ratio between the makespans of the two algorithms. We can see that HLP-OLS outperforms HLP-EST, except for a few instances with the application \emph{potri}, with an improvement close to 8\% on average. 
Comparing HLP-OLS and HEFT we notice that, even if the two algorithms have similar performances, HEFT is on average outperformed by HLP-OLS by 2\%, with a maximum of 60\% of improvement for HLP-OLS with some instances of \emph{potri}.
Moreover, HEFT has a significantly worse performance than HLP-OLS in strongly heterogeneous applications where there is a bigger perturbation in the (dis-)acceleration of the tasks on the GPU side, like \emph{forkJoin}, since in these irregular cases the allocation problem becomes more critical.\\

\begin{figure}
\centering
     \subfloat{%
       \includegraphics[width=.47\textwidth]{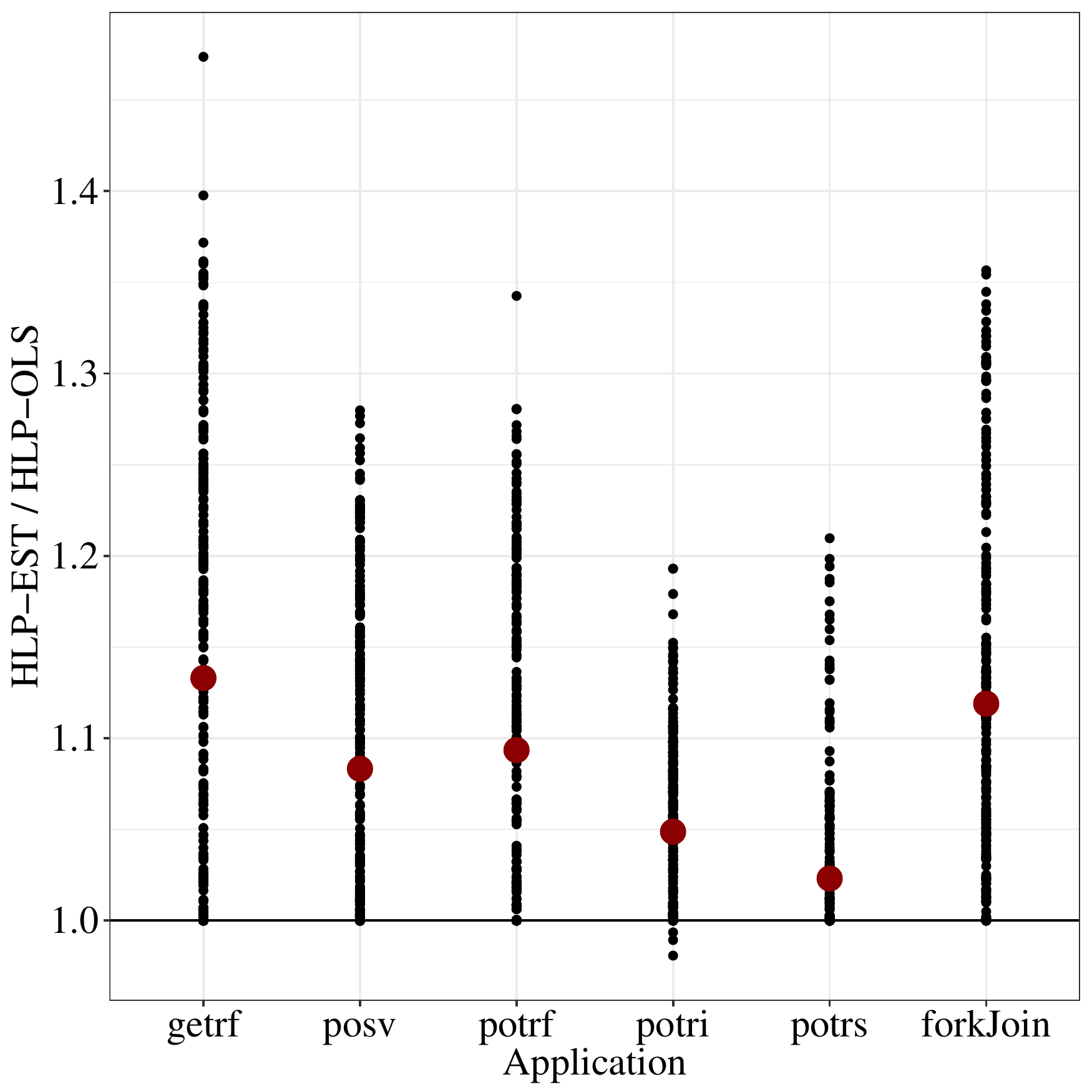}
     }
     \hfill
     \subfloat{%
       \includegraphics[width=.47\textwidth]{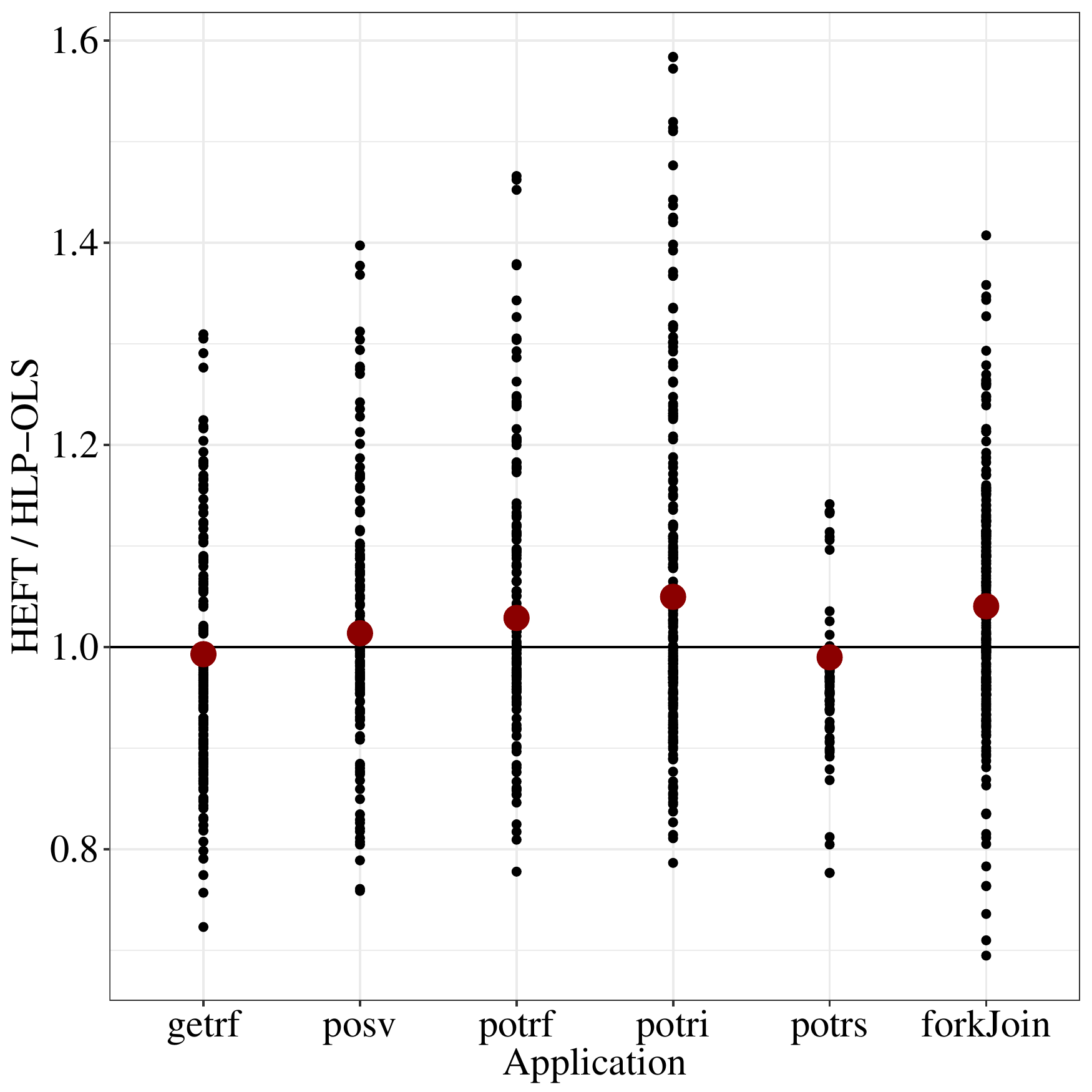}
     }
     \caption{Ratio between the makespans of HLP-EST and HLP-OLS (left), and HEFT and HLP-OLS (right) for each instance, grouped by application.}
     \label{fig:plotComp5}
\end{figure}

\noindent\textbf{Results for 3 resource types.}\\
To study the performance of the 3 algorithms we computed the ratio between each makespan and the optimal solution $LP^*$ of the linear program QHLP, which corresponds to a good lower bound of the optimal makespan.
Fig.~\ref{fig:plotR3} (left) shows the ratio of each instance of application and configuration. Notice that the red / bigger dot represents the mean value of the ratio for each application. We can see that QHLP-EST is on average outperformed by the two other algorithms. We also observe that even if QHEFT presents many outliers for the applications \emph{getrf}, \emph{posv} and \emph{potrf}, the algorithms outperforms on average QHLP-OLS.

\begin{figure}
\centering
\captionsetup{justification=centering}
\subfloat{
    \includegraphics[width=.47\textwidth]{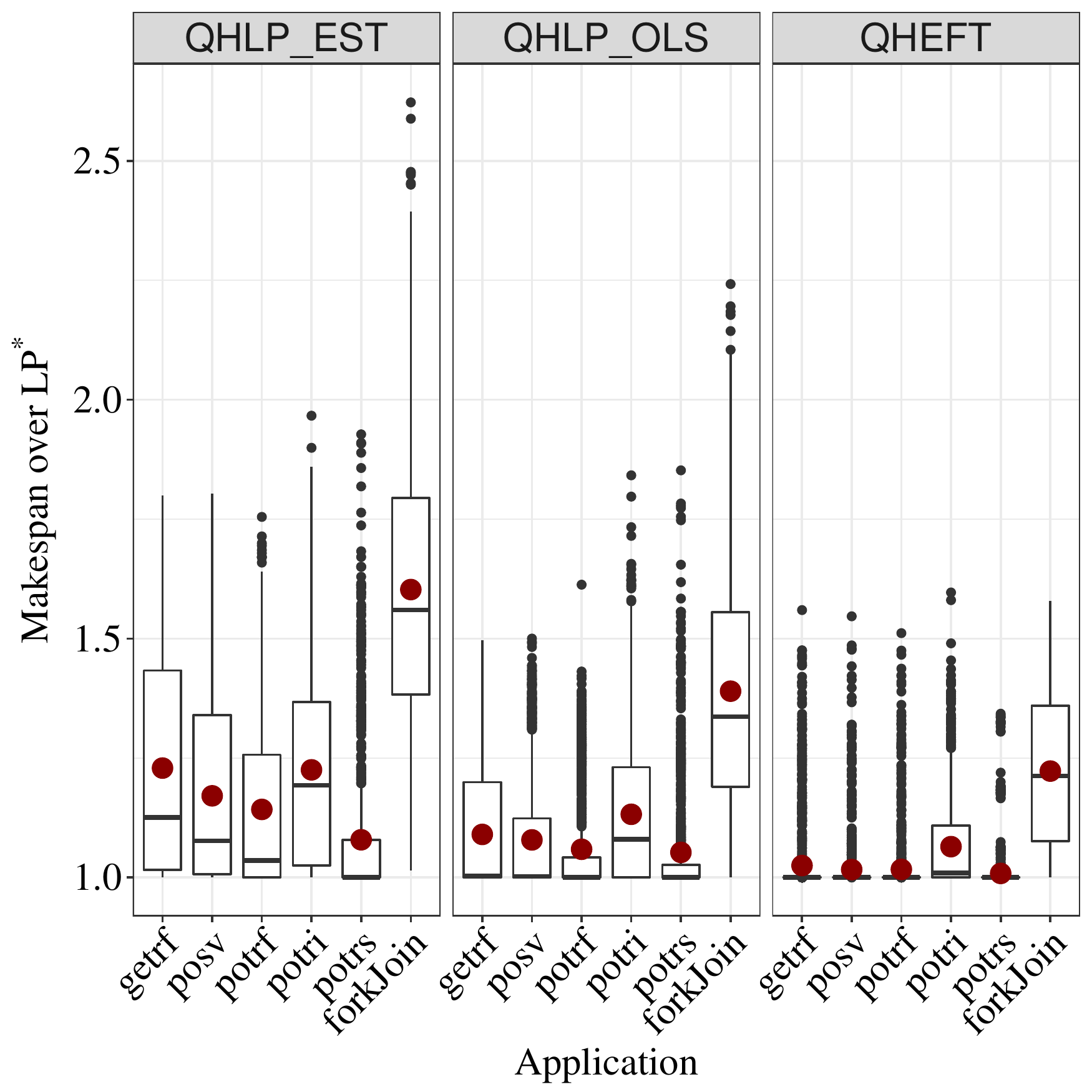}
}
\hfill
\subfloat{%
    \includegraphics[width=.47\textwidth]{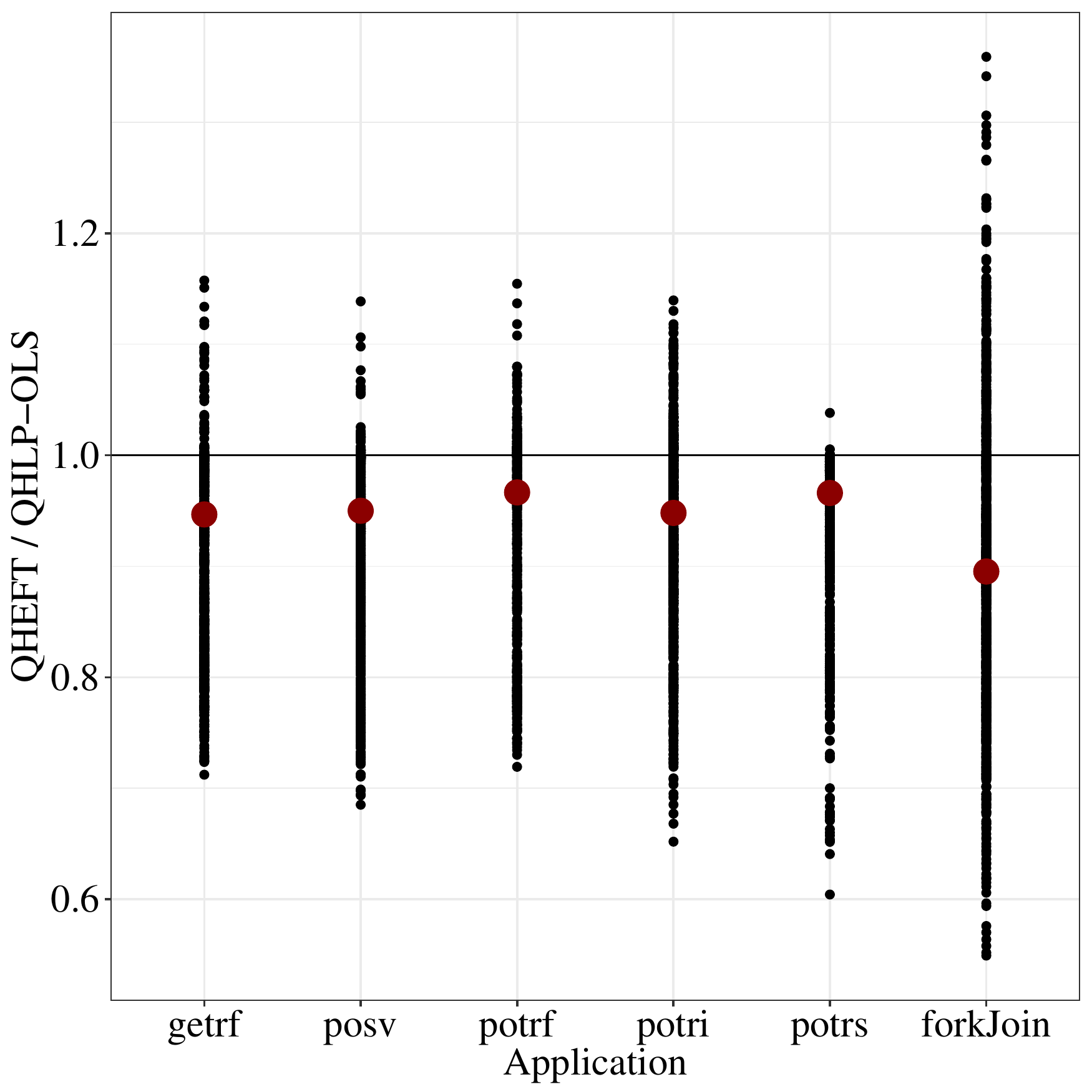}
}
\caption{Ratio of makespan over $LP^*$ for each instance (left), grouped by application for the algorithms generalized for 3 resource types. Ratio between the makespan of QHEFT and QHLP-OLS for each instance (right), grouped by application.}
\label{fig:plotR3}
\end{figure}

Fig.~\ref{fig:plotR3} (right) compares more specifically QHEFT and QHLP-OLS, by showing the ratio between the makespans of the two algorithms.
Comparing QHLP-OLS and QHEFT, we can see that QHEFT presents an improvement over QHLP-OLS of 5\% on average.
We also observe that the ratios for the most irregular application, \emph{fork-join}, are spread with instances favorable to QHEFT (up to 45\% of improvement) and other instances favorable to QHLP-OLS (up to 36\% of improvement).
\todoC{I don't know what conclusions can we take out of these observations}
Similar results were observed between QHLP-EST and QHLP-OLS as for 2 resource types and thus are not showed here.\\

Finally, we notice that the approximation ratios, computed with a lower bound of the optimal makespan, do not exceed 2 and thus are far from the theoretical bounds of the algorithms, even for the case with 3 resource types.

\subsection{On-line setting}

\noindent\textbf{Algorithms.}\\
We compared the performance, in terms of makespan, of the algorithm ER-LS (Section~\ref{ssec:onlineSetting}) with 3 baseline algorithms: \emph{EFT}, which schedules a task on the processor which give the earliest finish time for that task; \emph{Greedy}, which allocates a task on the processor type which has the smallest processing time for that task; and \emph{Random}, which randomly assigns a task to the CPU or GPU side. For the algorithms Greedy and Random, we used a List Scheduling algorithm to schedule the tasks once the allocations has been made.
The algorithms were implemented in Python (v. 2.8.6).
For the machine configurations, we used the same sets of pairs as for the off-line setting with 16, 32, 64 and 128 CPUs, and 2, 4, 8 and 16 GPUs for a total of 16 machine configurations.

We executed the algorithms only once with each combination of application and machine configuration since all algorithms are deterministic, except Random.
The running times of the algorithms were similar and took at most 5 seconds each for the biggest instances of applications.\\

\noindent\textbf{Results.}\\
Fig.~\ref{fig:plotON} (left) compares the ratios between the makespan of each of the on-line algorithms and $LP^*$. Due to large differences between the performances of Random and the 3 other algorithms, we kept only the algorithms ER-LS, EFT and Greedy.
Results show that Greedy is on average outperformed by ER-LS and EFT, and that EFT creates less outliers than the 2 other algorithms.
Fig.~\ref{fig:plotComp7} compares more specifically Greedy and ER-LS (left), and EFT and ER-LS (right), by showing the ratio between the makespans of the two algorithms.
We can see that ER-LS outperforms Greedy on average, with a maximum for the potri application where ER-LS performs 11 times better than Greedy. More specifically, there is an improvement of between 8\% and 36\% on average for ER-LS depending on the application considered, except for \emph{potrs} whose makespans are on average 10\% greater than for Greedy.

\begin{figure}
\centering
     \subfloat{%
       \includegraphics[width=.47\textwidth]{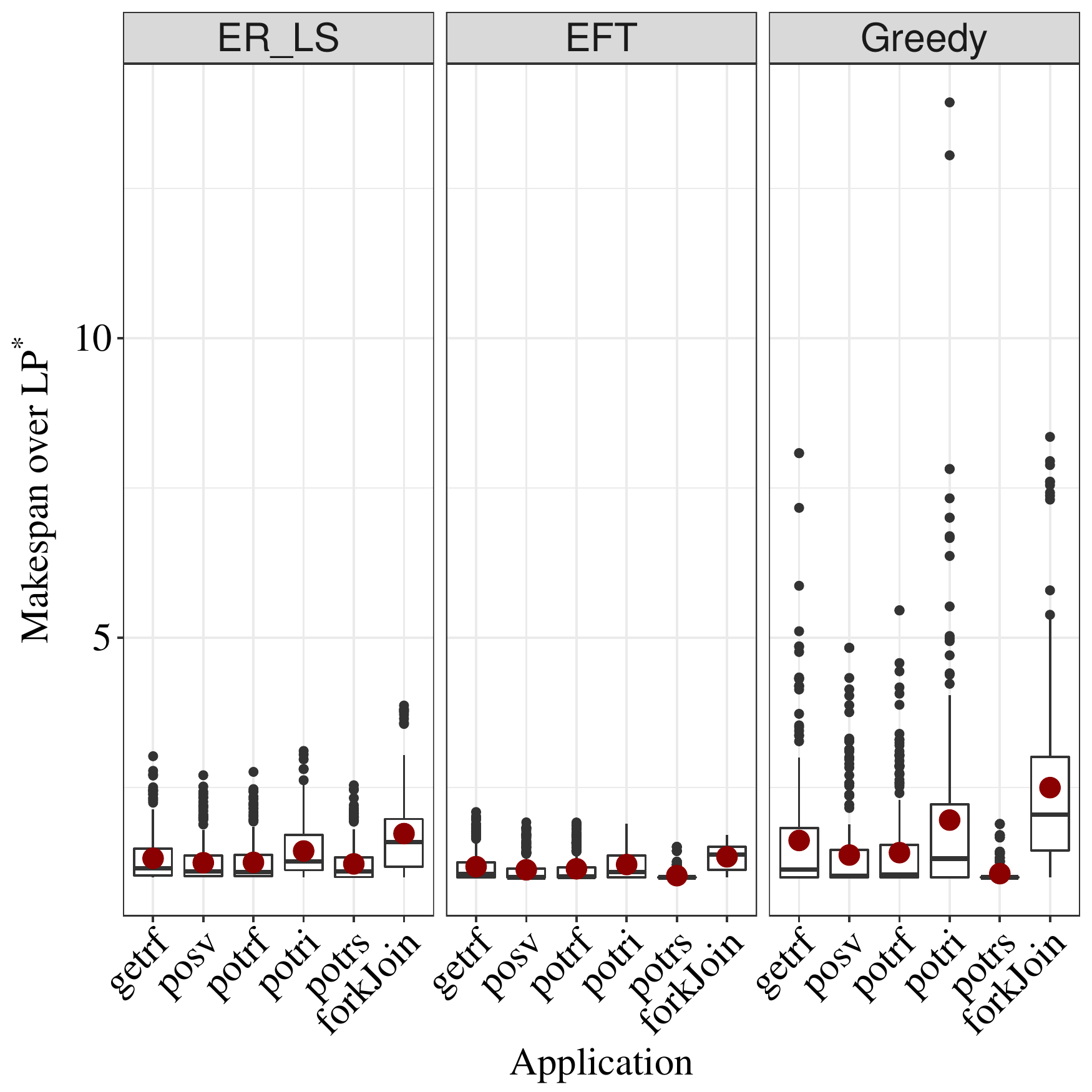}
     }
     \hfill
     \subfloat{%
       \includegraphics[width=0.47\textwidth]{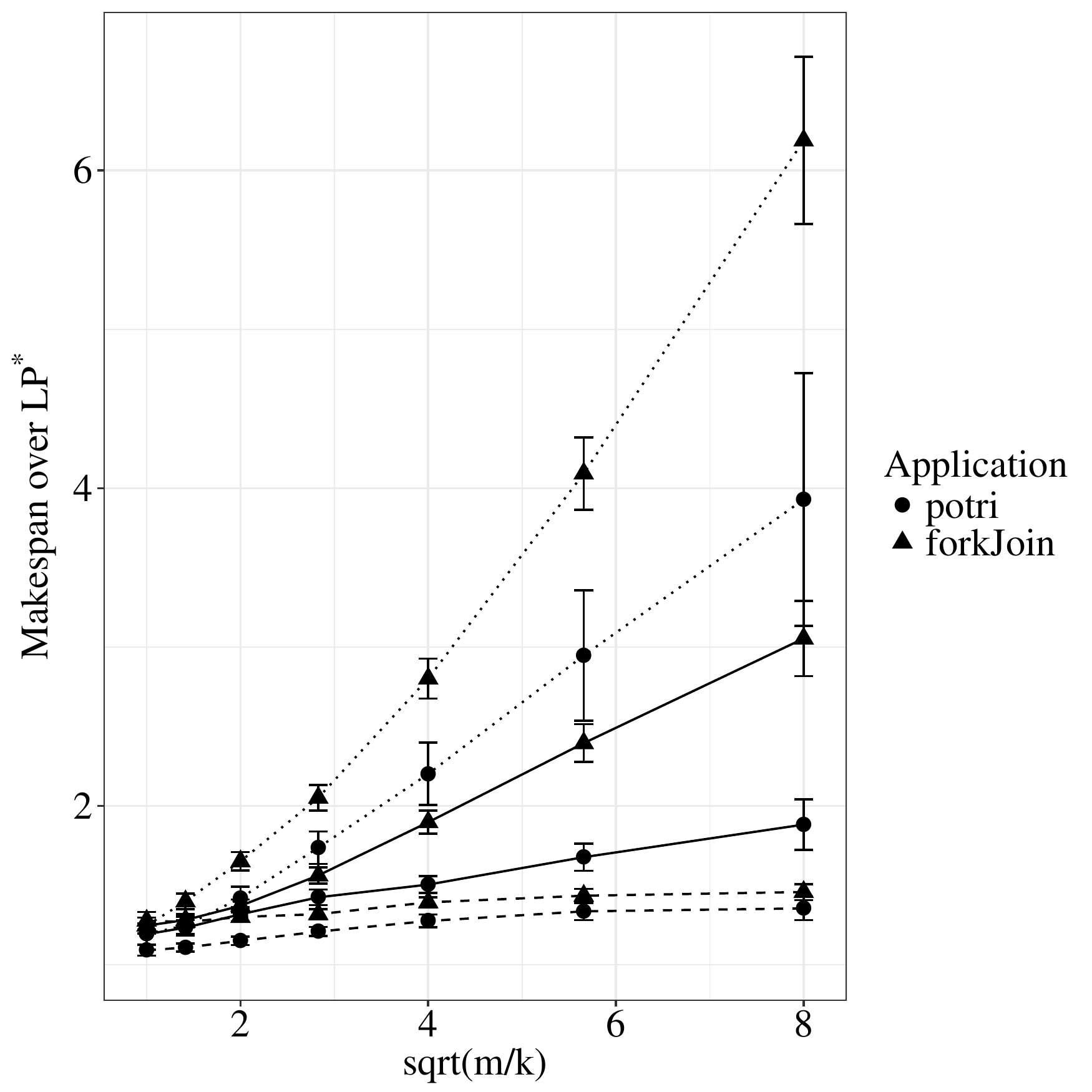}
     }
     \caption{Ratio of makespan over $LP^*$ for each instance, grouped by application for the on-line algorithms with 2 resource types (left). Mean competitive ratio, and standard error, of ER-LS (plain), EFT (dashed) and Greedy (dotted) as a function of $\sqrt{\frac{m}{k}}$ (right).}
     \label{fig:plotON}
\end{figure}

Comparing EFT and ER-LS, we can see that on average ER-LS is outperformed by EFT with a decrease of 11\% on average, and up to 60\% for certain instances of fork-join.
However, the worst-case competitive ratio for EFT can be directly obtained from the proof of the worst-case approximation ratio for HEFT, presented in Section~\ref{sec:LowerBounds}, by ordering the list of tasks of the instance by decreasing order of their rank. The competitive ratio of EFT is then at least $\frac{m+k}{k^2}\left(1-\frac{1}{e^k}\right)$, which is larger than the competitive ratio of ER-LS, namely $4\sqrt{\frac{m}{k}}$, when $k \leq \sqrt[3]{m}$.\\
\todoC{see if this argument is really needed}

\begin{figure}
\centering
     \subfloat{%
       \includegraphics[width=.45\textwidth]{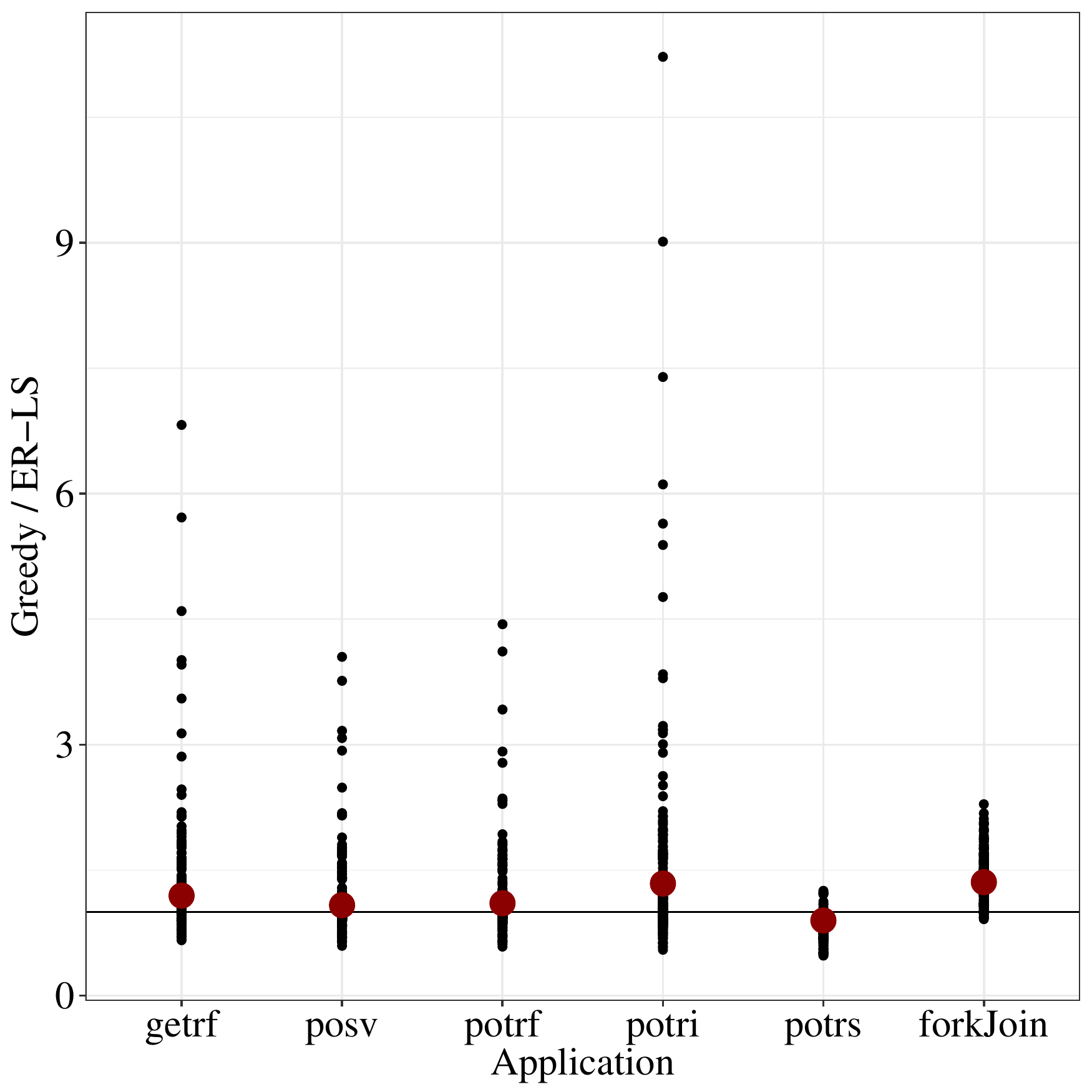}
     }
     \hfill
     \subfloat{%
       \includegraphics[width=0.45\textwidth]{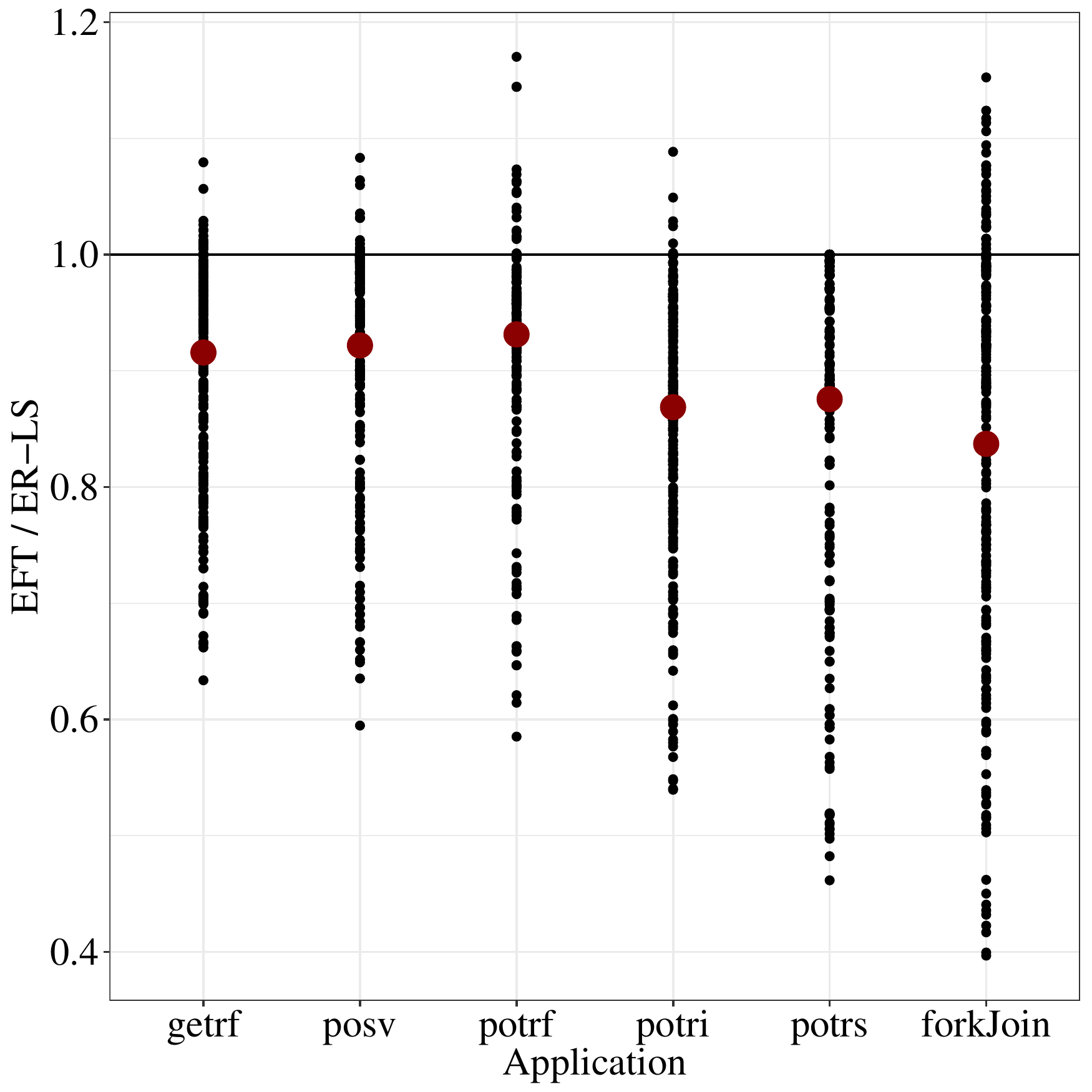}
     }
     \caption{Ratio between the makespans of Greedy and ER-LS (left), and EFT and ER-LS (right) for each instance, grouped by application.}
     \label{fig:plotComp7}
\end{figure}

We also study the performance of the 3 algorithms with respect to the theoretical upper bound given in Section~\ref{ssec:onlineSetting}.

Fig.~\ref{fig:plotON} (right) shows the mean competitive ratio of ER-LS, EFT and Greedy along with the standard error as a function of $\sqrt{\frac{m}{k}}$ associated to each instance.
To simplify the lecture, we only present the applications \emph{potri} and \emph{fork-join}, since other \emph{Chameleon} applications showed similar results. 
We observe that the competitive ratio is smaller than $\sqrt{\frac{m}{k}}$ and far from the theoretical upper bound of $4\sqrt{\frac{m}{k}}$ for ER-LS.

\vspace*{-0.13cm}

\section{Conclusions}\label{sec:Conclusions}
We studied the problem of scheduling parallel applications, represented by a precedence task graph, on heterogeneous multi-core machines.
We focused on generic approaches, non depending on the particular application, by distinguishing the allocation and the scheduling phases and proposed efficient algorithms with worst-case performance guarantees for both off-line and on-line settings.
In the off-line case, motivated by new lower bounds on the performance of existing algorithms, we refined the scheduling phase of the best known approximation algorithm and we presented a new algorithm that preserves the approximation ratio and performs better in our experiments.
We also extended this methodology for the more general case where the architecture is composed of $Q\geq 2$ types of resources. 
In the on-line case, we presented a $\Theta(\sqrt{\frac{m}{k}})$-competitive algorithm based on adequate rules. This is the first on-line result dealing with precedence constraints on hybrid machines.

From the practical point of view, an extensive simulation campaign on representative benchmarks constructed by real applications showed that it is possible to outperform the classical HEFT algorithm keeping reasonable running times for the case with 2 resource types. With more resource types, results showed that a simple generalization of an algorithm could present similar performances than HEFT.

For the on-line case, the algorithm based on rules is a good trade-off since it delivers a solution close to the optimal while keeping a running time similar to pure greedy algorithms.
We aim to implement it on a real run-time system (such as StarPU~\cite{Augonnet:StarPU}) which currently uses HEFT on successive sets of independent tasks. 

In this work we assumed that the communications between CPUs, GPUs and the shared memory are neglected.
Our next step is to introduce communication costs in the algorithms, which should not be too hard in both integer program and greedy rules.

\subsubsection*{Acknowledgments}\label{sec:acknowledgments}
This work was partially supported by FAPESP (S\~ao Paulo Research Foundation, grant \#2012/23300-7) and ANR Moebus Project.

All the experiments were performed on the Froggy platform of the CIMENT infrastructure (https://ciment.ujf-grenoble.fr), which is supported by the Rh\^{o}ne-Alpes region (GRANT CPER07$\_$13 CIRA) and the Equip@Meso project (reference ANR-10-EQPX-29-01) of the program ``Investissements d'Avenir'' supervised by the French Research Agency (ANR).

We also would like to thank Inria project team MORSE, especially Florent Pruvost and Mathieu Faverge, for the collaboration in this work.

\bibliographystyle{plain}

\end{document}